\documentclass[a4paper,UKenglish,cleveref, autoref, thm-restate]{lipics-v2021}

\hideLIPIcs  
\nolinenumbers

\title{A Tight Double-Exponential Lower Bound for High-Multiplicity Bin Packing}

\author{Klaus Jansen}{Kiel University}{kj@informatik.uni-kiel.de}{https://orcid.org/0000-0001-8358-6796}{}

\author{Felix Ohnesorge}{Kiel University}{foh@informatik.uni-kiel.de}{https://orcid.org/0009-0003-8023-3380}{}
\author{Lis Pirotton}{Kiel University}{lpi@informatik.uni-kiel.de}{https://orcid.org/0009-0001-4984-3696}{}

\authorrunning{K. Jansen, F. Ohnesorge and L. Pirotton}

\Copyright{K. Jansen, F. Ohnesorge and L. Pirotton} 



\category{} 

\relatedversion{} 

\funding{Funded by the Deutsche Forschungsgemeinschaft (DFG, German Research Foundation) - Project number 453769249.}

\acknowledgements{We thank Alberto Del Pia and Timo Berthold for discussions on constraint linearization and Stefan Weltge for pointing us to the references~\cite{DBLP:journals/mp/KaibelW15a,DBLP:books/daglib/0090562}. The idea for this work originated during a visit of Klaus Jansen to the MPI - INF Saarbrücken in 2024.}

\EventEditors{}
\EventNoEds{2}
\EventLongTitle{}
\EventShortTitle{}
\EventAcronym{}
\EventYear{}
\EventDate{}
\EventLocation{}
\EventLogo{}
\SeriesVolume{}
\ArticleNo{}
\usepackage[utf8]{inputenc}
\usepackage[T1]{fontenc}
\usepackage{todonotes}
\definecolor{babyblueeyes}{rgb}{0.63, 0.79, 0.95}

\usepackage{graphicx} 
\usepackage{amsfonts}
\usepackage{environ}
\usepackage{mathtools}
\usepackage{xspace}
\usepackage{csquotes}
\usepackage{algorithm}
\usepackage{algpseudocode}
\usepackage{acronym}

\acrodef{3-SAT}[\textsc{3-SAT}]{}
\acrodef{Bin Packing}[\textsc{Bin Packing}]{}
\acrodef{ILP}[\textsc{ILP}]{Integer Linear Program}
\acrodef{ETH}[\textsc{ETH}]{Exponential Time Hypothesis}

\newcommand{\vect}[1]{\boldsymbol{#1}}
\newcommand{\ZZ}{\mathbb{Z}}
\newcommand{\satbase}{\gamma}
\newcommand{\lisM}{M}
\newcommand{\linsymb}{\faStar[regular]}
\newcommand{\nlinsymb}{\faStar}

\newcommand{\ntag}{\refstepcounter{equation}\tag{\nlinsymb~\theequation}}
\newcommand{\ineqtag}{\refstepcounter{equation}\tag{\linsymb~\theequation}}

\newcounter{Ceq}
\renewcommand{\theCeq}{C\arabic{Ceq}}

\usepackage{environ}
\makeatletter
\NewEnviron{constraints}{%
    \begingroup
    \let\c@equation\c@Ceq
    \renewcommand{\theequation}{\theCeq}%
    \begin{align}
        \BODY
    \end{align}
    \endgroup
}
\makeatother

\crefname{claim}{Claim}{Claims}

\usepackage{xcolor}

\begin{document}

\maketitle

\begin{abstract}
    Consider a high-multiplicity \textsc{Bin Packing} instance $I$ with $d$ distinct item types.
    In 2014, Goemans and Rothvoss gave an algorithm with runtime ${{|I|}^2}^{O(d)}$ for this problem~[SODA'14], where $|I|$ denotes the encoding length of the instance $I$.
    Although Jansen and Klein~[SODA'17] later developed an algorithm that improves upon this runtime in a special case, it has remained a major open problem by Goemans and Rothvoss~[J.ACM'20] whether the doubly exponential dependency on $d$ is necessary.

    We solve this open problem by showing that unless the Exponential Time Hypothesis (\textsc{ETH}) fails, there is no algorithm solving the high-multiplicity \textsc{Bin Packing} problem in time ${{|I|}^2}^{o(d)}$.
    To prove this, we introduce a novel reduction from \textsc{3-SAT}.
    The core of our construction is efficiently encoding all information from a \textsc{3-SAT} instance with $n$ variables into an \textsc{ILP} with $O(\log n)$ variables and constraints.

    This result confirms that the Goemans and Rothvoss algorithm is essentially best-possible for \textsc{Bin Packing} parameterized by the number $d$ of item sizes in the context of XP time algorithms.
\end{abstract}

\section{Introduction}
The \acs{Bin Packing} problem is a classic optimization problem with many applications.
\begin{definition}[\acs{Bin Packing}]
    Given are \(d \in \ZZ_{>0}\) item types with sizes \(\vect{s} = (s_1, \ldots, s_d) \in (0, B]^d\) and multiplicities \(\vect{a} = (a_1, \ldots, a_d) \in \ZZ_{>0}^d\).
    The \acs{Bin Packing} problem asks to find the minimum number of bins of size \(B \in \mathbb{Z}_{>0}\) to pack all items.
\end{definition}
The \acs{Bin Packing} problem is also known as the (1-dimensional)
\textsc{Cutting Stock} problem, and its study goes back to the classical paper by Gilmore and Gomory~\cite{gilmore1961linear}.
While the problem is strongly NP-hard in general, a major research direction has focused on parameterized algorithms for the \emph{high-multiplicity} setting, where \(d\) is assumed to be a small parameter.
A breakthrough result in this area came
in 2014 from Goemans and Rothvoss~\cite{DBLP:journals/jacm/GoemansR20}, who proved that it is solvable in polynomial time for constant $d$. This answered an open question posed by McCormick, Smallwood and Spieksma \cite{DBLP:journals/mor/McCormickSS01},
as well as by Eisenbrand and Shmonin \cite{DBLP:journals/orl/EisenbrandS06}.

To obtain their result, Goemans and Rothvoss~\cite{DBLP:journals/jacm/GoemansR20} study the more general \textsc{Cone and Polytope Intersection} problem, defined as follows:
Given two polytopes \(\mathcal{P}, \mathcal{Q} \subseteq \mathbb{R}^d\): Is there a point in \(\mathcal{Q}\) that can be expressed as a non-negative integer combination of integer points in \(\mathcal{P}\)?
They gave an algorithm for this feasibility problem with time complexity \({{|\mathcal{P}|}^2}^{O(d)} \cdot |\mathcal{Q}|^{O(1)}\), where \(|\mathcal{R}|\) denotes the encoding length of a polytope \(\mathcal{R} \in \{\mathcal{P}, \mathcal{Q}\}\).
Additionally, they show how to construct a \textsc{Cone and Polytope Intersection} instance from any \textsc{Bin Packing} instance.
In this transformation, \(\mathcal{P} = \{ {\vect{x} \choose 1} \in \mathbb{R}^{d+1}_{\ge 0} | \vect{s}^\intercal \vect{x} \le B\} \) (the knapsack polytope)  contains all possible \emph{configurations} \(\vect x\) (i.e., multiplicity vectors of items that fit into a single bin), and \(\mathcal{Q} = \{\vect{a}\} \times [0,k]\) is constructed to encode the target item vector \(\vect{a}\) and the number of bins \(k\).
Using binary search over $k$, this yields an algorithm for \acs{Bin Packing} with runtime \({{|I|}^2}^{O(d)} \), where \(|I|\) denotes the encoding length of the instance. For $d=O(1)$ the encoding length is \(|I| = O(\log(\Delta))\), where $\Delta = \max\{\|\vect{a}\|_\infty, \|\vect{s}\|_\infty, B\}$.

This result was later improved by Jansen and Klein~\cite{DBLP:journals/mor/JansenK20}.
They gave an algorithm with time complexity \({{|V|}^2}^{O(d)} \cdot \log(\Delta)^{O(1)}\), where \(V\) is the set of vertices of the corresponding integer knapsack polytope.
This result improves upon the algorithm in~\cite{DBLP:journals/jacm/GoemansR20} if the number of vertices \(|V|\) is small.
Since $|V| \ge d + 1$, this gives a Fixed-Parameter Tractable (\textsc{FPT}) algorithm
parameterized by the number of vertices in the integer knapsack polytope.
On the other hand, the number of vertices can be bounded only by $|V| = O(\log \Delta)^{2^{O(d)}}$ \cite{DBLP:journals/combinatorica/CookHKM92,hartmann88}.
Therefore, the algorithm by Jansen and Klein has a worst case running time $O(\log \Delta)^{2^{O(d)}}$ which is identical to the running time of the algorithm by Goemans and Rothvoss.

Goemans and Rothvoss \cite{DBLP:journals/jacm/GoemansR20}  wrote in their journal paper:
\begin{quote}
    "A natural open problem that arises from this work is whether the double exponential running time is necessary."
\end{quote}
In addition to this question, they asked whether \acs{Bin Packing} can be solved in \textsc{FPT} time \(f(d) \cdot O(\log(\Delta))^{O(1)}\), where \(f(d)\) is an arbitrary function.

Recent work has highlighted the inherent complexity related to this parameterization.
Kowalik, Lassota, Majewski, Pilipczuk, and Sokolowski~\cite{DBLP:conf/sosa/KowalikLMPS24} studied the \textsc{Point in Cone} problem, a special case of the \textsc{Cone and Polytope Intersection} problem where the second polytope $\mathcal{Q}$ consists of just one point $q$.
They showed that a double exponential dependency on $d$ is unavoidable under the \ac{ETH} when $\mathcal{P}$ has an exponential number of inequalities. This results in an \acs{ETH}-tight lower bound.


\begin{definition}[\acs{ETH}, \cite{DBLP:conf/focs/ImpagliazzoPZ98}]
    The \ac{ETH} conjectures that \acs{3-SAT} cannot be solved in subexponential time. Specifically, there exists a constant \( \delta > 0 \) such that  no algorithm can solve \acs{3-SAT} with \(n\) variables in time \( 2^{\delta n} \).
\end{definition}
As proven in~\cite{DBLP:journals/jcss/ImpagliazzoPZ01}, this implies that there is no algorithm for \acs{3-SAT} with running time $2^{o(n+m)}$, where \(m\) denotes the number of clauses in the formula; see also
Theorem 14.4 in~\cite{DBLP:books/sp/CyganFKLMPPS15}.

Furthermore, the structure of solutions to the \acs{Bin Packing} problem with $d$ item sizes is known to be complex. Eisenbrand and Shmonin \cite{DBLP:journals/orl/EisenbrandS06} proved via an elegant combinatorial argument that there is always an optimal solution for \acs{Bin Packing} with a support (the number of distinct configurations needed) bounded
by $2^d$.
Recently, Jansen, Pirotton, and Tutas~\cite{DBLP:conf/esa/JansenPT25} showed that the support of any optimal solution in a \acs{Bin Packing} instance can be exponential in \(d\).
This structural hardness provides further evidence that a double exponential runtime may be optimal.

\subsection{Our Contribution}
We answer the open problem above by confirming that the double exponential runtime in Goemans and Rothvoss~\cite{DBLP:journals/jacm/GoemansR20} is necessary for \acs{Bin Packing}, assuming the \acs{ETH}.
Thus, their algorithm is asymptotically optimal regarding \textsc{XP} time algorithms parameterized by \(d\).
\begin{restatable}{theorem}{mainThm}
    \label{thm:main-result}
    There is no algorithm solving high-multiplicity \acs{Bin Packing} with \(d\) distinct item sizes in time \({{|I|}^2}^{o(d)}\), unless the \acs{ETH} fails.
\end{restatable}

We note that this does not answer whether \acs{Bin Packing} is \textsc{FPT} with respect to \(d\).
However, our result shows that any algorithm with runtime \(f(d) \cdot {|I|}^{O(1)}\) must have \(f(d) = {2^2}^{\Omega(d)}\).

To achieve this result, we reduce \acs{3-SAT} with $n$ variables to \acs{Bin Packing} using an \ac{ILP} formulation.
While introducing new techniques that might also be interesting to achieve other lower bounds, we make use of known results and adapt them for our purpose.
We now briefly describe the key components of our reduction.

First, we demonstrate that we may assume each variable in our \acs{3-SAT} instance occurs exactly twice positively and once negatively.
This technique is an extension of~\cite{DBLP:journals/dam/Tovey84} where only the number of occurrences of a variable is restricted to three, regardless of literal polarity.
Next, we encode the instance in a single integer $Z$ with bounded size, similar to~\cite{DBLP:journals/mp/JansenKL23,DBLP:journals/siamdm/JansenLL16,DBLP:journals/mp/KaibelW15a,DBLP:conf/stoc/MandersA76,DBLP:books/daglib/0090562}. In contrast to the prior results, our encoding enables us to efficiently extract information from this number. More precisely, using just $O(\log n)$ \ac{ILP} variables and constraints, we are able to extract the clause numbers in which a variable appears. This idea stems from a discussion with Weltge. In their work~\cite{DBLP:journals/mp/KaibelW15a}, Kaibel and Weltge study lower bounds on sizes of \acp{ILP} without additional variables.

Altogether, we construct an \ac{ILP} that only allows a certain set of solutions. These solutions are carefully constructed as they will correspond to specific \acs{Bin Packing} configurations that, with a certain target vector, can model a feasible assignment to \acs{3-SAT}. Some of the constraints are already given in~\cite{DBLP:conf/esa/JansenPT25}. However, their \ac{ILP} is not able to model \acs{3-SAT}. To achieve this property, we have to ensure that a single \acs{3-SAT} variable cannot be set both negatively and positively at the same time.
While other reductions~\cite{DBLP:journals/mp/JansenKL23,DBLP:conf/sosa/KowalikLMPS24,DBLP:books/daglib/0090562} also use \ac{ILP} encodings, we manage to achieve a more compact encoding where the number of variables and equality constraints is only \emph{logarithmic} in \(n\) and the \ac{ILP} allows \(O(n) = O(2^d)\) specific solutions. For instance, the number of inequality constraints in the reduction from \textsc{Subset Sum with Multiplicities} to \textsc{Point in Cone} by~\cite{DBLP:conf/sosa/KowalikLMPS24} is linear in \(n\) (and exponential in \(d\)).

To transform our \ac{ILP} into a \acs{Bin Packing} instance, we use the aggregation technique in~\cite{DBLP:conf/esa/JansenPT25}.
Then, we determine the satisfiability of \acs{3-SAT} by solving a family of \(O(n^4)\) \textsc{Bin Packing} instances, denoted as \(\mathcal{I}^{\texttt{BP}}(\hat{\vect \chi})\).
Here \(\vect{\hat{\chi}}\) is a vector containing additional, in polynomial time computable, information of a valid \acs{3-SAT} solution; for example, the component $\hat{\chi}_3$ represents how many variables are set to \texttt{false}.

The compact encoding of \acs{3-SAT} and the construction of feasible solutions allow us to translate the \(2^{o(n)}\) lower bound for \acs{3-SAT} into the desired \({{|I|}^2}^{o(d)}\) lower bound for \acs{Bin Packing}.
It is worth noting that we require the sizes of the \acs{Bin Packing} instance to be encoded in binary due to large item sizes.
We firmly believe that this reduction technique, particularly the flexible encoding via an \ac{ILP}, is of independent interest.
It demonstrates a powerful pattern for establishing lower bounds that can likely be adapted to prove similar results for other problems, such as for high multiplicity \acp{ILP} with few constraints, $m$-dimensional knapsack, multiple knapsack, scheduling problems, and high multiplicity block structured $n$-fold and $2$-stage \acp{ILP}.
Other examples of a double exponential lower bounds under the \acs{ETH} are given in~\cite{DBLP:journals/siamcomp/CyganPP16,DBLP:journals/talg/FominGLSZ19,
    HunkenschroderKKLL25,
    DBLP:journals/mp/JansenKL23,DBLP:journals/toct/KnopPW20,
    DBLP:conf/sosa/KowalikLMPS24,DBLP:journals/jacm/KunnemannMSSW25,DBLP:conf/icalp/MarxM16}.

\section{Preliminaries}
Before presenting our main result, we establish the notation, definitions, and concepts used throughout the paper.
For $n\in \mathbb{Z}_{\ge1}$, we define $[n] \coloneqq \{1, 2, \ldots, n\}$ and $[n]_0 \coloneqq \{0,1, \ldots, n-1\}$.
For a vector \(\vect{x} \in \mathbb{Z}^n\), we denote its components by \(x_1, x_2, \ldots, x_n\).
We also use a convenient notation for vectors.
For example, if \(\vect{x} = (x_1, \ldots, x_n) \in \ZZ^n\) and \(a,b,c \in \ZZ\), then \(\vect{y} = (\vect{x}, a, b, c)\) denotes the vector \((x_1, \ldots, x_n, a, b, c)\).
As our reduction is from \acs{3-SAT}, we formally define this problem here:
\begin{definition}[\acs{3-SAT}]
    Given $n \in \mathbb{Z}_{\ge0}$ boolean variables $v_i, i \in [n]_0$ and a boolean formula $\varphi = C_0 \land C_1 \land \cdots \land C_{m-1}, m \in \mathbb{Z}_{\ge0}$, where each clause $C_j$ consists of at most three literals, e.g., $C_j = (\ell_{j1} \lor \ell_{j2} \lor \ell_{j3})$. A literal $\ell_{jk}$ is either a variable $v_i$ or its negation $\lnot v_i$ for some $i \in [n]_0$. The \acs{3-SAT} problem asks whether there exists an assignment $\phi : \{v_0, \ldots, v_{n-1}\} \mapsto \{\texttt{true},\texttt{false}\}$ that satisfies the formula $\varphi$.
\end{definition}

The concept of a \emph{well-structured} \acs{3-SAT} instance is central to many reductions in complexity theory, especially when aiming for tight lower bounds.
Tovey~\cite{DBLP:journals/dam/Tovey84} introduced a transformation that converts any \acs{3-SAT} instance into an equivalent one with at most \(3\) occurrences per variable.
We use a slight extension of this transformation to obtain instances where each variable appears exactly twice positively and once negatively.
Similar transformations have also been used in other reductions~\cite{DBLP:journals/dam/BermanKS07,DBLP:journals/tcs/JansenM95}.

\begin{restatable}[\faCut, Well-Structured \acs{3-SAT}; Extension of \cite{DBLP:journals/dam/Tovey84}]{lemma}{structuredSat}
    \label{def:structured-3sat}
    Given any instance of \acs{3-SAT} with \(n\) variables and \(m = O(n)\) clauses, there exists an equivalent instance with \(n' = O(n)\) variables and \(m'=O(n)\) clauses, where each variable \(v_i\) has exactly two positive appearances in the clauses and exactly one negative appearance.
\end{restatable}
\begin{proof}
    We provide the proof in \Cref{sec:proof-structuredSat} for completeness.
\end{proof}

In this work, we construct a non-trivial \ac{ILP}.
Whenever a constraint is indexed by a solid star (e.g., \ref{eq:lis-constraint}), it indicates that the constraint is nonlinear.
Similarly, a non-solid star (\linsymb) denotes a constraint that is linear but not yet in its equality form.
We state these purely for readability.
By applying standard techniques, such as adding slack variables or using the conversions captured in the next lemma, we transform them into linear equality constraints.
The full set of final equality constraints is listed in \Cref{sec:equalityConst}, where each unstarred constraint shares the exact same number as its starred counterpart in the main text (e.g., the final form of \ref{eq:lis-constraint} is \ref{eq:lis-constraint-linear}).

Linearization of quadratic terms has been extensively studied in the context of $0-1$ quadratic programming (BQP) and quadratic integer programming (QIP) with the original work going back to~\cite{DBLP:journals/mp/McCormick76}.
Quadratic terms $x \cdot y$ with Boolean and integer variables $x,y$ can be replaced using additional variables and/or inequalities; see also \cite{fortet1960applications,DBLP:journals/ior/GloverW74,DBLP:journals/mp/McCormick76}. For example,
if $x,y$ are both Boolean variables, then $z= x \cdot y$ can be replaced by
\begin{align*}
    z \le x,\ z \le y,\  z \ge x+y-1, \text{ and }  z \ge 0.
\end{align*}
If $x \in \{0,1\}$ and $y \in [L,U]$ and $y$ integral, then $z = x \cdot y$ can be replaced by
\begin{align*}
    z \le y,\ z \le U \cdot x,\ z \ge L \cdot x,\text{ and } z \ge y + (x-1) \cdot U.
\end{align*}
We show how to linearize such terms without the introduction of new variables under certain conditions.
The fact that we can linearize these constraints without introducing additional variables enables us to reason using equivalent nonlinear constraints.

\begin{restatable}[\faCut]{lemma}{constraintLinearization}
    \label{lem:linearization}
    The nonlinear equation \(y = \sum_{j = 1}^k{x_j \cdot \chi_j}\) involving integer variables \(x_j\) with known bounds \(0 \leq x_j \leq U\) and binary variables \(\chi_j \in \{0, 1\}\) with \(\sum_{j = 1}^k \chi_j \leq 1\) can be equivalently expressed using \(O(k)\) linear inequalities and no additional variables.
\end{restatable}
\begin{proof}
    We provide the proof in \Cref{sec:proof-lemLinearization}.
\end{proof}

Another key concept at the core of our reduction is the \acs{ILP} aggregation.
For this, we have to transform each linear inequality constraint into one equality constraint through the introduction of slack variables.
While the general concept of an \acs{ILP} aggregation is not new, recently Jansen, Pirotton, and Tutas~\cite{DBLP:conf/esa/JansenPT25} presented a new technique allowing them to integrate upper bounds on variables into the aggregation.
We summarize this result:

\begin{lemma}[\cite{DBLP:conf/esa/JansenPT25}]
    \label{lem:aggregation}
    Let \(d, k = O(\log n)\). Consider an \ac{ILP} $A\vect{x}=\vect{b}$, with \(\vect{x} \in \ZZ_{\ge0}^{d}\), \(\vect{x} \le \vect{u}\) with $A = (a_{ij})_{i\in [k],j\in[d]} \in \ZZ^{k\times d}$, $\vect{u} \in \ZZ_{\ge0}^{d}$ and $\vect{b}\in \ZZ^{k}$ and let $\Delta \coloneqq \|A\|_\infty$ be the largest absolute value in $A$.
    The vector $\vect{x}$ is a feasible integer solution to the ILP $A\vect{x}=\vect{b}$, if and only if there exists a unique $\vect{y} \in \ZZ_{\ge0}^{d+1}$ such that $(\vect{x},\vect{y})$ is a feasible integer solution to
    \begin{equation}\label{eq:aggregated}
        \begin{array}{l}
            \sum_{i=1}^k \big(\lisM^{i-1} \sum_{j=1}^d (a_{ij} x_j)\big)
            + \sum_{j=1}^d \big(\lisM^{k+j-1} (x_j + y_j)\big)
            + \lisM^{k+d}\big(\sum_{j=1}^d (x_j+y_j)+y_{d+1}\big) \\
            = \sum_{i=1}^k (\lisM^{i-1} b_i)
            + \sum_{j=1}^d \big(\lisM^{k+j-1} u_j\big)
            + \lisM^{k+d} U,
        \end{array}
    \end{equation}
    where $U \coloneqq \sum_{j=1}^d u_j$ and $\lisM \coloneqq \Delta U + \max(\|\vect{b}\|_ \infty, \|\vect{u}\|_\infty) + \Delta + 2$.
\end{lemma}
\begin{proof}[Proof Sketch]
    Given an ILP of the form $A\vect{x}=\vect{b}, \vect{x} \in \ZZ_{\ge0}^d, \vect{x} \le \vect{u}$, we first replace the external upper bounds $\vect{x} \le \vect{u}$ by $d$ constraints $x_j + y_j = u_j, \forall j \in [d]$, while introducing non-negative integer slack variables $y_j \in \mathbb{Z}_{\ge0}$.
    Additionally, we add a constraint that upper bounds the upper bounds, i.e.\ $\sum_{j=1}^d u_j +y_{d+1} = U$, with $U \coloneqq \sum_{j=1}^d u_j$ and $y_{d+1} \in \mathbb{Z}_{\ge0}$.
    Next, we define the large base number $\lisM \coloneqq \Delta U + \max(\|\vect{b}\|_\infty, \|\vect{u}\|_\infty) + \Delta + 2$ that prevents carries when the constraints are aggregated.
    Now, we multiply each constraint by a power of $\lisM$, i.e., we multiply the first constraint by 1, the second one by $\lisM$, the third by $M^2$ and so on.
    Finally, we sum up all weighted equations to \Cref{eq:aggregated}.

    Since $\lisM$ is sufficiently large, each constraint can be seen as a single base-$\lisM$ integer. Therefore, a feasible solution $\vect x$ to the original \ac{ILP} is also a feasible solution to the aggregated one (while adding the unique slack variables) and the two systems are equivalent.
\end{proof}

For the full proof and the equivalence of both
\acp{ILP} (i.e., $\vect{x}$ is a feasible integral solution to $A\vect{x} = \vect{b}$, $0\le \vect{x} \le \vect{u}$, if and only if $(\vect{x},\vect{y})$ is a feasible integer solution with $\vect{x},\vect{y} \ge 0$ to \Cref{eq:aggregated}), we refer to Section 3 in \cite{DBLP:conf/esa/JansenPT25}.
Note that this \acs{ILP} aggregation is different to the technique in~\cite{DBLP:journals/jcss/JansenKMS13} where they reduce \textsc{c-Unary Bin Packing}, i.e., \textsc{Unary Bin Packing} with $c$ dimensions to \textsc{Unary Bin Packing} with one dimension. Here the authors prove that \textsc{Unary Bin Packing} is W[1]-hard, parameterized by the number of bins by first reducing \textsc{Subgraph Isomorphism} to \textsc{10-Unary Bin Packing} and then using a reduction to \textsc{Unary Bin Packing}.
This approach relies on sophisticated methods based on \emph{k-non-averaging} sets and \emph{sumfree} sets.
Our reduction heavily relies on the extension of the \ac{ILP} constructed in~\cite{DBLP:conf/esa/JansenPT25}, thus we capture the main properties of their construction in the following lemma:
\begin{lemma}[\faCut, \cite{DBLP:conf/esa/JansenPT25}]
    \label{lem:ilp-jpt}
    For any \(\gamma > \ZZ_{\geq 2}\), there exists an \acs{ILP} formulation \(A\vect{x} = \vect{b}\), with exactly \(n\) unique solutions that can be computed in polynomial time and where the first \(\log(n)+1\) coordinates of each solution form the following set:
    \[
        X = \left\{\left(x_1^{\texttt{bin}}, x_2^{\texttt{bin}}, \ldots, x_{\log n}^{\texttt{bin}}, \gamma^{\sum_{\ell=1}^{\log n}{2^{\log(n)-\ell} \cdot x_{\ell}^{\texttt{bin}}}}\right) \Big| x_\ell^{\texttt{bin}} \in \{0, 1\}\right\},
    \]
    with \(\vect x, \vect b \in \ZZ_{\geq 0}^{O(\log n)}\), and \(\|A\|_\infty \leq {\gamma^2}^{\log n}\).
\end{lemma}
\begin{proof}[Proof Sketch]
    The above set of solutions can be represented by the following set of nonlinear constraints with \(\tilde{r}_0 = 1\):
    \begin{constraints}
        \ntag\label{eq:lis-constraint}
        \tilde{r}_{\ell} & = \tilde{r}_{\ell-1} \left(1 + (\gamma^{2^{\log(n) - \ell}}-1)x^{\texttt{bin}}_\ell\right) & \forall \ell \in [\log n]
    \end{constraints}
    For the full list of linear constraints with the desired properties and the proof of correctness, we refer to \Cref{sec:equalityConst} (denoted as \eqref{eq:lis-constraint-linear}) and the original work~\cite{DBLP:conf/esa/JansenPT25}, respectively.
\end{proof}

\section{Reduction from \acs{3-SAT} to \acs{Bin Packing}}
We aim to prove via reduction from \acs{3-SAT} to \acs{Bin Packing} that there is no algorithm with a runtime of \({|I|}^{2^{o(d)}}\) for \acs{Bin Packing}, unless the \ac{ETH} fails.

Consider an arbitrary \acs{3-SAT} instance with \(n\) variables and \(m\) clauses.
We apply a series of simplifying transformations.
First, by the Sparsification Lemma~\cite{DBLP:conf/focs/ImpagliazzoPZ98} we can assume that \(m = O(n)\).
Next, we apply \Cref{def:structured-3sat}, such that each variable now appears positively in exactly two clauses and negatively in exactly one clause.
This transformation runs in polynomial time and does not significantly increase the number of variables and clauses, so we still have \(m = O(n)\).
Finally, to simplify the notation, we assume w.l.o.g.\ that the number of variables \(n\) is a power of two.
If \(n\) is not \(2^k\) for some integer \(k\), we repeatedly add a new variable \(v_{n}\) and the trivially true clause \((v_{n} \lor v_{n} \lor \lnot v_{n})\) until the number of variables is a power of two.
Now denote this well-structured \acs{3-SAT} instance as \(\mathcal{I}^{\texttt{SAT}}\).

In preparation for the reduction, we encode the \acs{3-SAT} instance \(\mathcal{I}^{\texttt{SAT}}\) in a single large integer \(Z\).
Similar techniques have been used in other reductions~\cite{DBLP:journals/mp/JansenKL23,DBLP:journals/siamdm/JansenLL16,DBLP:journals/mp/KaibelW15a,DBLP:conf/stoc/MandersA76,DBLP:books/daglib/0090562}; however, our encoding is structurally simpler than previous approaches.
Furthermore, we efficiently extract specific information from \(Z\); to our knowledge, such a method has not been previously utilized.

The general idea is to construct a large (base \(\satbase^m\)) integer with three dedicated "digits" for each variable \(v_i\) that encode the clauses in which \(v_i\) appears positively and negatively.
To that end, let \(j, k \in [m]_0\) be the clauses where \(v_i\) appears positively and let \(\ell \in [m]_0\) be the clause where \(v_i\) appears negatively.
We define \(C_{pos1}^{(i)} \coloneqq \satbase^{j}\), \(C_{pos2}^{(i)} \coloneqq \satbase^{k}\), and \(C_{neg\vphantom{p2}}^{(i)} \coloneqq \satbase^{\ell}\), with $\satbase\in\mathbb{Z}_{>3}$.
Later, we will set \(\gamma = 4n+1\); however, for the following lemmas this is not required.

\begin{lemma}
    \label{lem:sat-encoding}
    A well-structured \acs{3-SAT} instance with \(n\) variables and \(m\) clauses can be represented as an integer \(Z\) of size at most \(\satbase^{3nm}\) with \(\satbase \in \mathbb{Z}_{>3}\), such that
    \begin{align}\label{eq:Z}
        Z = \sum_{i=0}^{n-1}\left( (\satbase^m)^{3i} \cdot C_{pos1}^{(i)} + (\satbase^m)^{3i+1} \cdot C_{pos2}^{(i)} + (\satbase^m)^{3i+2} \cdot C_{neg\vphantom{p2}}^{(i)}\right).
    \end{align}
    The integer $Z$ can be computed in polynomial time.
\end{lemma}
\begin{proof}
    First, we verify that \(Z \leq \satbase^{3nm}\).
    The expression for \(Z\) can be viewed as a number represented in base \(\satbase^m\) with \(3n\) digits.
    We obtain the following upper bound for \(Z\)
    \begin{align*}
        Z \leq \sum_{i=0}^{3n-1} (\satbase^m-1)\cdot(\satbase^m)^i = (\satbase^m-1) \sum_{i=0}^{3n-1} (\satbase^m)^i = (\satbase^m-1) \frac{(\satbase^m)^{3n}-1}{\satbase^m-1} < \satbase^{3nm}.
    \end{align*}
    From the definition of a well-structured \acs{3-SAT} instance, we know that each variable \(v_i, i \in [n]_0\) appears in exactly three clauses: Twice positively (in the clauses $j,k$) and once negatively (in clause $\ell$).
    Note that these indices can be found in polynomial time by iterating through all clauses.
    Defining \(C_{pos1}^{(i)} := \satbase^{j}, C_{pos2}^{(i)} := \satbase^{k}, C_{neg\vphantom{p2}}^{(i)} := \satbase^{\ell}\) can be done by scanning over all clauses for all \(i \in [n]_0\).

    Next, we compute \(Z\) from these terms.
    The expression for \(Z\) can be computed efficiently using \(O(n)\) multiplications and \(O(n)\) additions.
    Let \(k\) be the bit-length of \(\satbase^m\) and the \(C^{(i)}\) terms.
    The computation involves a sequence of \(O(n)\) multiplications where the intermediate sum's bit-length grows linearly.
    The \(j\)-th multiplication for \(j \in \{1, \ldots, O(n)\}\) multiplies an intermediate sum of \(O(j \cdot k)\) bits by \(\satbase^m\).
    Using naive \(O(k^2)\) multiplication, this step takes \(O((j \cdot k)k)\) time.
    The total sum for this summation is the sum of all \(O(n)\) steps, i.e.,
    \begin{align*}
        \sum_{j=1}^{O(n)}O(j \cdot k^2) = O(k^2)\sum_{j=1}^{O(n)}j = O(k^2 \cdot n^2).
    \end{align*}
    Substituting \(k = O(m \log \satbase)\), this results in \(O(n^2m^2\log^2\satbase)\).
    This completes the proof.
\end{proof}
Note that the time complexity can be improved significantly by using the \(O(k \log
k)\) time multiplication algorithm by Harvey and van der Hoeven~\cite{harvey2021integer} and/or exploiting the fact that all involved terms are a power of two.

\subsection{Construction of the \acs{ILP}}
The core of our reduction is an \ac{ILP} formulation with a carefully crafted solution structure.
Later, each solution of this \ac{ILP} will correspond to one of five configuration types for the \acs{Bin Packing} problem. These configuration types represent different ways a \acs{3-SAT} variable can satisfy clauses.
The high-level idea is as follows:
Suppose, for each \(i \in [n]_0\) we have to select exactly two of the following five solution types in \(\mathcal{X}^{(i)}\) such that the sum of all selected solutions equals a target vector \(\vect{t}\).
\begin{align*}
    \mathcal{X}^{(i)} = \left\{
    \begin{pmatrix}
        C^{(i)}_{pos1} \\ \gamma^i \\ 0\\
    \end{pmatrix},
    \begin{pmatrix}
        C^{(i)}_{pos2} \\ \gamma^i \\ 0
    \end{pmatrix},
    \begin{pmatrix}
        C^{(i)}_{neg\vphantom{p2}} \\ 2 \cdot \gamma^i \\ 0
    \end{pmatrix},
    \begin{pmatrix}
        0 \vphantom{C^{(i)}_{pos1}} \\ \gamma^i\\ 1
    \end{pmatrix},
    \begin{pmatrix}
        0 \vphantom{C^{(i)}_{pos1}} \\ 0\\ 1
    \end{pmatrix}
    \right\}, \quad\quad \vect t = \begin{pmatrix}
                                       \sum_{j=0}^{m-1}{\gamma^j} \\ \sum_{i=0}^{n-1}{2 \cdot \gamma^i} \\ 2n - m
                                   \end{pmatrix}
\end{align*}

Recall that \(C_{pos1}^{(i)}, C_{pos2}^{(i)}, C_{neg\vphantom{p2}}^{(i)}\) are of the form \(\satbase^j\), thus when \(\satbase\) is sufficiently large, the first coordinate of \(\vect{t}\) ensures that each clause \(j \in [m]_0\) is selected (satisfied) exactly once.
The last two coordinates ensure a consistent assignment of the variables.
We formally show this later.

In the remainder of this work, given an integer \(i \in [n]_0\), we refer to the binary encoding of \(i\) as \(\vect{x^{\texttt{bin}}} = (x^{\texttt{bin}}_1, \ldots, x^{\texttt{bin}}_{\log n})\), with \(i = \sum_{\ell = 1}^{\log n} 2^{\log(n)-\ell} \cdot x_{\ell}^{\texttt{bin}}\).

The foundation of our \ac{ILP} construction is the formulation in \Cref{lem:ilp-jpt}, which uses \(O(\log n)\) variables and constraints to create \(n\) distinct solution vectors. Specifically, for each \(i \in [n]_0\) there exists a unique solution to \eqref{eq:lis-constraint} where the variable \(\tilde{r}_{\log n} = \satbase^i\) and the variables \(x_1^{\texttt{bin}}, \ldots, x_{\log n}^{\texttt{bin}}\) represent the binary encoding of \(i\).

We extend the \ac{ILP} given by \eqref{eq:lis-constraint} as described in the following theorem:

\begin{theorem}
    \label{thm:possible-solutions}
    Given a well-structured \acs{3-SAT} encoding \(Z\) as defined in \Cref{eq:Z}, there exists an \ac{ILP} formulation \(A\vect x \leq \vect b\) with variable vector \(\vect x = (\vect\alpha, \vect\beta, \vect\chi)^\intercal \in \ZZ_{\geq 0}^{O(\log n)}\)
    and right-hand side \(\vect b \in \mathbb{Z}^{O(\log n)}\) that has exactly \(5n\) solutions.

    These solutions are partitioned into $n$ groups of $5$, indexed by $i \in [n]_0$. For each index $i$, the subvector $\vect\alpha$ takes exactly one of the following $5$ distinct assignments:
    $$
        \vect\alpha =
        \begin{pmatrix}
            \alpha_1\vphantom{C^{(i)}_{pos1}} \\ \alpha_2\vphantom{\gamma^i} \\ \alpha_3
        \end{pmatrix}\in
        \left\{
        \begin{pmatrix}
            C^{(i)}_{pos1} \\ \gamma^i \\ 0\\
        \end{pmatrix},
        \begin{pmatrix}
            C^{(i)}_{pos2} \\ \gamma^i \\ 0
        \end{pmatrix},
        \begin{pmatrix}
            C^{(i)}_{neg\vphantom{p2}} \\ 2 \cdot \gamma^i \\ 0
        \end{pmatrix},
        \begin{pmatrix}
            0\vphantom{C^{(i)}_{pos1}} \\ \gamma^i\\ 1
        \end{pmatrix},
        \begin{pmatrix}
            0\vphantom{C^{(i)}_{pos1}} \\ 0\\ 1
        \end{pmatrix}
        \right\}.
    $$
    For each index $i$, the corresponding subvector $\vect\beta \in \mathbb{Z}_{\geq 0}^{O(\log n)}$ is uniquely determined (depending solely on $i$ and $Z$) and is computable in polynomial time. Finally, the subvector $\vect\chi$ always satisfies $\vect\chi \in \{0, 1\}^{O(1)}$.
\end{theorem}

We achieve this result through the introduction of \(O(\log n)\) constraints, \(O(\log n)\) auxiliary variables \(\vect\beta\), and 4 binary variables \(\vect\chi = (\chi_1, \ldots, \chi_4)\).
For improved readability, we deviate from a strict component-wise notation (e.g. \(\vect\beta = (\beta_1, \beta_2, \ldots)\)) and will implicitly include all introduced auxiliary variables within the vector \(\vect\beta\), unless explicitly stated otherwise.
In particular, we include variables in constraint \eqref{eq:lis-constraint} in this vector  (e.g. \(\vect\beta = (\tilde{r}_0, \ldots,\tilde{r}_{\log n}, \vect x^{\texttt{bin}}, \ldots)\)).
The values of these variables are unique for fixed \(i, Z\), as shown in~\cite{DBLP:conf/esa/JansenPT25}.
For all newly introduced variables, we argue this property separately.

In order to construct the desired solutions for \(\vect\alpha\), the clauses relevant for the \acs{3-SAT} variable \(v_i\) (i.e., \(C^{(i)}_{pos1}, C^{(i)}_{pos2},C^{(i)}_{neg\vphantom{p2}}\)) need to be extracted from \(Z\).
We can achieve this efficiently since the variables \(x^{\texttt{bin}}_1, \ldots, x^{\texttt{bin}}_{\log n}\) already represent the binary encoding of \(i\), and thus describe the path through a binary search tree to reach the \(i\)-th block of \(Z\).
This idea could also be useful to compactly encode other problems in \ac{ILP} formulations.
The following constraints extract the corresponding block from \(Z\) into a variable \(z_{\log n}\).
Again, we list the equivalent linear constraints in \Cref{sec:equalityConst} and focus on the nonlinear version for improved readability.
An example procedure is illustrated in \Cref{fig:extraction}.
\begin{constraints}
    z_0     & = Z
    \label{eq:s_logn1}                                                                                            \\
    z_j     & = q_j \cdot (\satbase^m)^{3n/2^{j+1}} + r_j                                      & \forall j \in [\log n]_0
    \label{eq:s_logn2}                                                                                            \\
    r_j     & \leq (\satbase^m)^{3n/2^{j+1}} -1                                                & \forall j \in [\log n]_0
    \ineqtag\label{eq:s_logn3}                                                                                            \\
    z_{j+1} & = q_j \cdot x^{\texttt{bin}}_{j+1} + r_j \cdot (1 - x^{\texttt{bin}}_{j+1}) & \forall j \in [\log n]_0
    \ntag\label{eq:s_logn4}
\end{constraints}

For the proofs of \Crefrange{claim:blocks}{claim:abcd}, we refer the reader to~\Cref{sec:proof-thm-claims}.
\begin{restatable}[\faCut]{claim}{claimBlocks}
    \label{claim:blocks}
    The constraints given in \Crefrange{eq:s_logn1}{eq:s_logn4} ensure that \(z_{\log n}\) contains exactly the block of \(Z\) corresponding to variable \(v_i\), i.e., for $i \in [n]_0$ we get \(z_{\log n} = C^{(i)}_{pos1} \cdot (\satbase^m)^0 + C^{(i)}_{pos2} \cdot (\satbase^m)^1 + C^{(i)}_{neg\vphantom{p2}} \cdot (\satbase^m)^2\). Additionally, the values of all introduced variables are unique for fixed $i$ and $Z$.
\end{restatable}

\begin{figure}[t]
    \centering
    \includegraphics[width=\columnwidth]{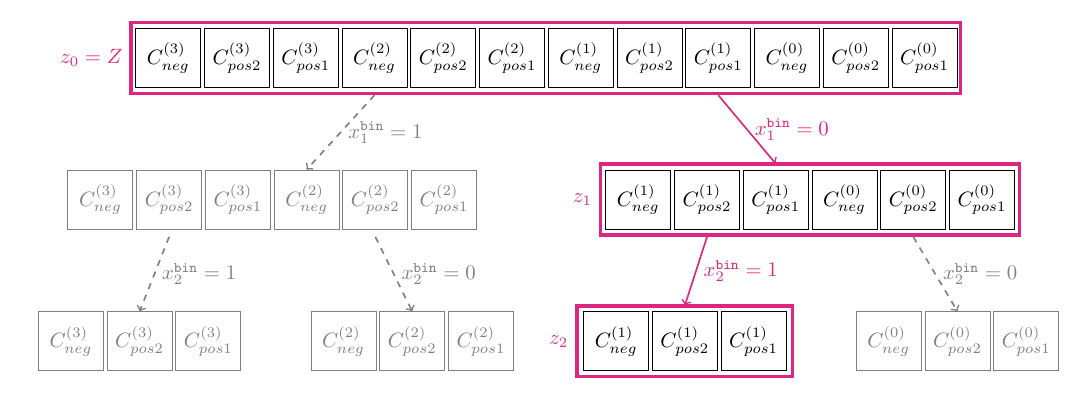}%
    \caption{Extraction of the block corresponding to variable \(v_1\) (\(\vect x^{\texttt{bin}} = (0, 1)\)) with Constraints \eqref{eq:s_logn1} to \eqref{eq:s_logn4}.}
    \label{fig:extraction}
\end{figure}

Extracting the values \(C^{(i)}_{pos1}, C^{(i)}_{pos2}, C^{(i)}_{neg\vphantom{p2}}\) into auxiliary variables \(c_{pos1}, c_{pos2}, c_{neg}\) can be done with the following constraints:
\begin{constraints}
    z_{\log n}           & = q_c \cdot \satbase^m + c_{pos1}
    \label{eq:cpscn1}                              \\
    q_c                   & = c_{neg} \cdot \satbase^m + c_{pos2}
    \label{eq:cpscn2}                              \\
    c_{pos1}, c_{pos2}, c_{neg} & \leq \satbase^m-1
    \ineqtag\label{eq:cpscn3}
\end{constraints}

\begin{restatable}[\faCut]{claim}{claimCs}
    \label{claim:cs}
    The constraints given in \Crefrange{eq:s_logn1}{eq:cpscn3} ensure that the variables $c_{pos1},c_{pos2}$ and $c_{neg}$ match the corresponding values in the definition of $Z$ in \Cref{eq:Z} for a given \(i\), i.e., $c_{pos1} = C^{(i)}_{pos1},c_{pos2}=C^{(i)}_{pos2}$ and $c_{neg}=C^{(i)}_{neg\vphantom{p2}}$. Additionally, the values of all introduced variables are unique for fixed $i$ and $Z$.
\end{restatable}

Using these auxiliary variables, we define the solution values for \(\vect\alpha\), as required by \Cref{thm:possible-solutions}.
This construction utilizes binary variables \(\alpha_3, \chi_1, \chi_2, \chi_3, \chi_4 \in \{0, 1\}\) and integer variables \(\alpha_1, \alpha_2\):

\begin{constraints}
    \alpha_1     & = c_{pos1} \cdot \chi_1 + c_{pos2} \cdot \chi_2 + c_{neg} \cdot \chi_3
    \ntag\label{eq:alpha1}                                                                                                                                             \\
    \alpha_2     & = \tilde{r}_{\log n} \cdot \chi_1 + \tilde{r}_{\log n} \cdot \chi_2 + 2 \cdot \tilde{r}_{\log n} \cdot \chi_3 + \tilde{r}_{\log n} \cdot \chi_4
    \ntag\label{eq:alpha2}                                                                                                                                             \\
    1 - \alpha_3 & = \chi_1 + \chi_2 + \chi_3
    \label{eq:alpha3}                                                                                                                                                  \\
    \chi_4       & \leq \alpha_3
    \ineqtag\label{eq:alpha4-linear}
\end{constraints}

\begin{restatable}[\faCut]{claim}{claimAbcd}
    \label{claim:abcd}
    For \(i \in [n]_0\), the constraints given in \Crefrange{eq:s_logn1}{eq:alpha4-linear} imply that
    \begin{align*}
        \begin{pmatrix}
            \alpha_1 \vphantom{C^{(i)}_{pos1}} \\ \alpha_2\vphantom{\gamma^i} \\ \alpha_3
        \end{pmatrix}\in
        \left\{
        \begin{pmatrix}
            C^{(i)}_{pos1} \\ \gamma^i \\ 0\\
        \end{pmatrix},
        \begin{pmatrix}
            C^{(i)}_{pos2} \\ \gamma^i \\ 0
        \end{pmatrix},
        \begin{pmatrix}
            C^{(i)}_{neg\vphantom{p2}} \\ 2 \cdot \gamma^i \\ 0
        \end{pmatrix},
        \begin{pmatrix}
            0\vphantom{C^{(i)}_{pos1}} \\ \gamma^i\\ 1
        \end{pmatrix},
        \begin{pmatrix}
            0\vphantom{C^{(i)}_{pos1}} \\ 0\\ 1
        \end{pmatrix}
        \right\}.
    \end{align*}
    All other introduced variables are binary (part of \(\vect\chi\)).
\end{restatable}

As established in \Cref{claim:blocks,claim:cs,claim:abcd}, the desired properties for \(\vect \alpha\) (see \Cref{thm:possible-solutions}) can be enforced by a compact nonlinear system with \(O(\log n)\) variables and \(O(\log n)\) constraints.
Our goal is to linearize this system while preserving its compact size.
First, we address \eqref{eq:lis-constraint}, providing its equivalent linear formulation in \Cref{sec:equalityConst}.
For this specific subsystem, \cite{DBLP:conf/esa/JansenPT25} showed that all variables remain unique for any fixed \(i\).
This leaves the nonlinear constraints \eqref{eq:s_logn4}, \eqref{eq:alpha1}, and \eqref{eq:alpha2}.
We apply \Cref{lem:linearization} to linearize these constraints without introducing any additional variables.
Because no new variables are added during this process, the crucial uniqueness property established in \Cref{claim:blocks,claim:cs,claim:abcd} is preserved.
A complete list of these final linearized constraints is provided in \Cref{sec:linearConst}.
With this, we have proven \Cref{thm:possible-solutions}.

\subsection{Constructing the \acs{Bin Packing} Instance}
\label{sec:constructing-binpacking}

In the previous section, we showed how to construct a compact \ac{ILP}, allowing only very specific solution vectors (see \Cref{thm:possible-solutions}).
Now, we show how to construct a \acs{Bin Packing} instance from this \ac{ILP}.

First, we aggregate the \ac{ILP} from \Cref{thm:possible-solutions} into an \ac{ILP} with a single knapsack constraint.
To this end, we use the aggregation technique in \cite{DBLP:conf/esa/JansenPT25} that is stated in \Cref{lem:aggregation}.
The \ac{ILP} from \Cref{thm:possible-solutions} contains inequality constraints, in order to apply \Cref{lem:aggregation}, we first transform it into an \ac{ILP} with equality constraints by introducing slack variables \(\vect y \in \ZZ^{O(\log n)}\).
We list the final equality constraints \Crefrange{eq:lis-constraint-linear}{eq:alpha4-eq} in \Cref{sec:equalityConst} for the sake of completeness.
Observe that since we introduce at most one slack variable per inequality, the number of variables and constraints remains in \(O(\log n)\).

Applying the aggregation in \Cref{lem:aggregation} to the \ac{ILP} of form \(A\vect x = \vect b\), we obtain an \ac{ILP} with a single linear constraint of the form \(\vect s^\intercal\vect x = B\) where \(\vect x = (\vect \alpha, \vect \beta, \vect \chi, \vect y)^\intercal\).
Note that the aggregation adds \(O(\log (n))\) extra slack variables to the slack variable vector \(\vect y\).
Thus, we still have \(\vect s, \vect x \in \ZZ^{O(\log n)}\).
Recall, that \(\vect \alpha\) contains three variables and \(\vect \beta = (\vect x^\texttt{bin}, \tilde{\vect r}, \vect z, \vect q, \vect r, c_{pos1}, c_{pos2}, c_{neg}, \vect y, \ldots) \in \ZZ^{O(\log n)}\) is the variable vector containing all variables with unique values for any fixed \(i, Z\) (see \Cref{thm:possible-solutions}).

\begin{restatable}[\faCut]{lemma}{ilpAggrDim}
    \label{aggilp-dim}
    The aggregated \ac{ILP} of the form \(\vect s^\intercal\vect x = B\) for \(\vect x=(\vect \alpha, \vect \beta, \vect \chi, \vect y)\) has dimensions \(\vect s, \vect x \in \mathbb{Z}_{\geq 0}^{O(\log n)}\) and \(||\vect s||_{\infty}, B \leq \gamma^{O(n^2 \log n)}\).
    Furthermore, \(||\vect x||_{\infty} \leq \gamma^{O(n^2)}\) holds for every feasible solution \(\vect x\).
\end{restatable}
The above lemma is a direct application of the technique from~\cite{DBLP:conf/esa/JansenPT25} (summarized in \Cref{lem:aggregation}), therefore, we defer the proof of this lemma to \Cref{sec:aggregation}.

Note that the solution vectors satisfying the equation \(\vect s^\intercal\vect x = B\) are exactly those of the above constructed \ac{ILP} (those listed in \Cref{thm:possible-solutions}).

We are now ready to construct a family of \acs{Bin Packing} instances \(\mathcal{I}^{\texttt{BP}}(\hat{\vect \chi})\) where \(\hat{\vect \chi} \in \{0, \ldots, 2n\}^4\).
Each instance has \(d = O(\log n)\) item types, defined by the item size vector \(\vect s\) and bin capacity \(B\).
Intuitively, this implies that for any feasible configuration \(\vect x\), the inequation \(\vect s^\intercal\vect x \leq B\) holds.
We now set the item multiplicities \(\vect a\) such that a \acs{Bin Packing} solution is forced to select specific configurations, as we will prove later (\Cref{lem:linear-combination}).
We set \(\vect a \coloneqq (\hat{\vect \alpha}, \hat{\vect \beta}, \hat{\vect \chi}, \hat{\vect y})^\intercal\), where each component is a vector, to be defined in the following paragraph.
The first component is set to
\begin{align}\label{eq:alpha-hat}
    \hat{\vect \alpha} \coloneqq \left(\sum_{j=0}^{m-1}{\gamma^j}, \sum_{i=0}^{n-1}{2 \cdot \gamma^i}, 2n -m\right)^\intercal.
\end{align}
Intuitively, this ensures that any solution with \(2n\) bins is guaranteed to (1) select each clause, and (2) correspond to a consistent assignment of \acs{3-SAT} variables \(\{v_i\} \mapsto \{\texttt{true}, \texttt{false}\}\).
We formally prove this in \Cref{lem:linear-combination}.

Now, let \(\vect \beta(i)\) denote the unique variable values for each \ac{ILP}-variable in \(\vect \beta\).
Note that by \Cref{thm:possible-solutions} these values are unique for each \(i \in [n]_0\) and can be computed in polynomial time.
For example, consider the variables appearing in \eqref{eq:s_logn1}-\eqref{eq:s_logn4}.
For fixed \(Z\) and \(x^{\texttt{bin}}_i\), the values of \(q_i, r_i\), and \(z_i\) can be calculated by \(O(\log n)\) Euclidean divisions.
We then set
\begin{align}\label{eq:beta-hat}
    \hat{\vect \beta} \coloneqq \sum_{i=0}^{n-1}2 \cdot \vect \beta(i).
\end{align}

Let now \(\mathcal{S}\) be any set of \(2n\) solutions to \(\vect s^\intercal(\vect \alpha, \vect \beta, \vect \chi ,\vect y) = B\).
For any variable \(x_i\) (e.g., \(\alpha_i, c_{pos1}, \chi_i, y_i\)), we define \(\hat{x}_i\) as the sum of its values across all \(2n\) solutions in \(\mathcal{S}\): \(\hat{x}_i \coloneqq \sum_{s \in \mathcal{S}}x_i(s)\).
The vector \(\hat{\vect y}\) is composed of these summed values for each of the \(O(\log n)\) slack variables in \(\vect y\).
The value of each component can be calculated from the known values of the other variables (i.e., \(\hat{\vect \alpha}, \hat{\vect \beta}, \hat{\vect \chi}\)), since each constraint contains at most one slack variable.

We demonstrate how to calculate \(\hat{\vect y}\) with the example of component \(y^{cc}_1\) in the first constraint of \eqref{alpha1-eq}:
\begin{align*}
    \alpha_1 - c_{pos1} + U^{cc} \cdot \chi_1 + y_1^{cc} = U^{cc}.
\end{align*}
Here \(y^{cc}_1\) is the single slack variable and \(U^{cc}\) is a known upper bound. Then:
\begin{equation*}
    \begin{array}{crl}
                            & \sum_{s \in \mathcal{S}}\left(\alpha_1(s) - c_{pos1}(s) + U^{cc} \cdot \chi_1(s) + y_1^{cc}(s) \right) & = \sum_{s \in \mathcal{S}}U^{cc} \\
        \Longleftrightarrow & \hat{\alpha}_1 - \hat{c}_{pos1} + U^{cc} \cdot \hat{\chi}_1 + \hat{y}_1^{cc}                           & = 2n \cdot U^{cc}.
    \end{array}
\end{equation*}
Since \(\hat{y}_1^{cc}\) is the only unknown, it can be calculated.
We set \(\hat{\vect y}\) by solving each of the \(O(\log n)\) constraints this way.
Then, the multiplicities in instance \(\mathcal{I}^{\texttt{BP}}(\hat{\vect \chi})\) are given by: \(\vect a = (\hat{\vect \alpha}, \hat{\vect \beta}, \hat{\vect \chi}, \hat{\vect y})\).
Note that the above construction ensures \(\vect s^\intercal \vect a = 2n \cdot B\).
We also note that while it is sufficient for our analysis to consider all possibilities for \(\hat{\vect \chi}\), the number of considered vectors could be reduced significantly.
When considering the corresponding variables \(\vect \chi = (\chi_1, \chi_2, \chi_3, \chi_4)\) in the constructed \acs{ILP}, we notice that \(\hat{\chi_1}, \hat{\chi_2}\), and \(\hat{\chi_3}\) count the number of times a clause is selected.
Since we aim to choose all \(m\) clauses, \(\hat{\chi}_3 = m - \hat{\chi}_1 - \hat{\chi}_2\).
A similar argument can be made for \(\hat{\chi}_4\).
This variable counts the usage of the first \emph{slack} solution.
This number can be calculated from \(\hat{\chi}_3\).
These two arguments bring the number of considered vectors from \((2n+1)^4\) to \((2n+1)^2\) since only \(\hat{\chi}_1\) and \(\hat{\chi}_2\) need to be guessed.

The polytope defining all feasible configurations in the \acs{Bin Packing} instance is given by \(P = \{\vect x \in \ZZ_{\geq 0}^d | \vect s^\intercal\vect x \leq B\}\).
We prepare our final proof by showing that each solution with at most \(2n\) bins uses only configurations \(\vect x\) that satisfy \(\vect s^\intercal\vect x = B\).

\begin{claim}
    \label{claim:eqConfigs}
    Any solution to \(\mathcal{I}^{\texttt{BP}}(\hat{\vect \chi})\) with at most \(2n\) bins uses only configurations \(\vect x\) that satisfy \(\vect s^\intercal\vect x = B\).
\end{claim}
\begin{proof}
    First, note that by construction of \(\mathcal{I}^{\texttt{BP}}(\hat{\vect \chi})\), the total size of items is \(\vect s^\intercal \vect a = 2n \cdot B\).
    Now, assume for the sake of contradiction that there exists a solution to the constructed \acs{Bin Packing} instance with \(2n\) bins that uses a configuration \(\vect x\) with \(\vect s^\intercal \vect x < B\).
    Then, the remaining size of items is \(2n \cdot B - \vect s^\intercal \vect x > (2n - 1) \cdot B\) and thus cannot be packed into the remaining \(2n - 1\) bins.
    A contradiction.
\end{proof}

Now, define the set of feasible solutions to the constructed \ac{ILP} as
\begin{align*}
    \mathcal{X} = \left\{
    \vect{x_1^{(i)}}, \ldots, \vect{x_5^{(i)}} =
    \begin{pmatrix} C^{(i)}_{pos1} \\ \gamma^i \\ 0 \\ \vect\beta(i) \\ \vect \chi \\ \vect y \end{pmatrix},
    \begin{pmatrix} C^{(i)}_{pos2} \\ \gamma^i \\ 0 \\ \vect\beta(i) \\ \vect \chi \\ \vect y \end{pmatrix},
    \begin{pmatrix} C^{(i)}_{neg\vphantom{p2}} \\ 2\gamma^i \\ 0 \\ \vect\beta(i) \\ \vect \chi \\ \vect y \end{pmatrix},
    \begin{pmatrix} 0\vphantom{C^{(i)}_{pos1}} \\ \gamma^i \\ 1 \\ \vect\beta(i) \\ \vect \chi \\ \vect y \end{pmatrix},
    \begin{pmatrix} 0\vphantom{C^{(i)}_{pos1}} \\ 0 \\ 1 \\ \vect\beta(i) \\ \vect \chi \\ \vect y \end{pmatrix}
    \,\middle|\, i \in [n]_0
    \right\}.
\end{align*}
\Cref{claim:eqConfigs} implies that any \acs{Bin Packing} solution uses only configurations \(\vect x \in \mathcal{X}\).
We are ready to prove the final lemma, stating that the \acs{3-SAT} instance \(\mathcal{I}^{\texttt{SAT}}\) is solvable if and only if there exists a solution to any of the \acs{Bin Packing} instances \(\mathcal{I}^{\texttt{BP}}(\hat{\vect \chi})\).
An illustration of a constructed \acs{Bin Packing} instance is given in \Cref{fig:bin-packing}.
\begin{figure}[ht]
    \centering
    \includegraphics[width=0.9\columnwidth]{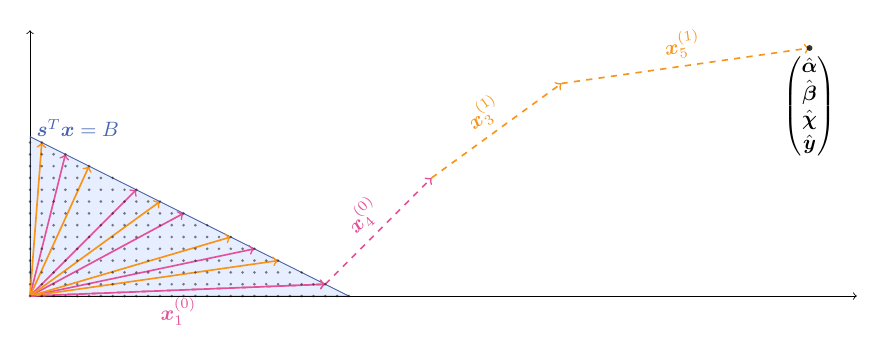}%
    \caption{Example for a \ac{3-SAT} instance with \(2\) variables \(v_0\) (pink) and \(v_1\) (yellow).
        The illustrated solution corresponds to the assignment \(v_0 = \texttt{true}\) and \(v_1 = \texttt{false}\).}
    \label{fig:bin-packing}
\end{figure}

\begin{lemma}
    \label{lem:linear-combination}
    Given an integer \(\gamma > 4n\), the well-structured \acs{3-SAT} instance \(\mathcal{I}^{\texttt{SAT}}\) is a \textsc{Yes}-Instance if and only if there exists a \(\hat{\vect \chi} \in \{0, \ldots, 2n\}^4\) such that \(\mathcal{I}^{\texttt{BP}}(\hat{\vect \chi})\) has a solution with at most \(2n\) bins.
\end{lemma}
\begin{proof}
    Let the well-structured \acs{3-SAT} instance \(\mathcal{I}^\texttt{SAT}\) be given as in \Cref{def:structured-3sat}.

    We first remember that due to the aggregation, there cannot be any carry-overs between the variable dimensions in the \acs{ILP}.
    Therefore, when solving the \acs{Bin Packing} instance, we ask how many configurations of each type to use i.e., we solve the linear combination
    \begin{align*}
        \sum_{i = 0}^{n-1}\sum_{k=1}^{5}{\lambda^{(i)}_{k} \vect{x^{(i)}_k}} = \vect a =
        \begin{pmatrix}
            \sum_{j=0}^{m-1}{\gamma^j}         \\
            \sum_{i=0}^{n-1}{2 \cdot \gamma^i} \\
            2n - m                             \\
            \hat{\vect \beta}                  \\
            \hat{\vect \chi}                   \\
            \hat{\vect y}
        \end{pmatrix}
        \text{, with } \sum_{i=0}^{n-1}\sum_{k=1}^{5}{\lambda^{(i)}_{k}} \leq 2n.
    \end{align*}

    For the remainder of this proof, we will refer to \(\vect{x^{(i)}_1}\), \(\vect{x^{(i)}_2}\), and \(\vect{x^{(i)}_3}\) as \emph{clause-paying} vectors and to \(\vect{x^{(i)}_4}\) and \(\vect{x^{(i)}_5}\) as \emph{slack} vectors.
    Similarly, we refer to \(\lambda^{(i)}_1\), \(\lambda^{(i)}_2\), and \(\lambda^{(i)}_3\) as \emph{clause-paying} coefficients.

    "$\Rightarrow$": We start with the "if" direction and assume that there exists a satisfying assignment \(\phi : \{v_0, \ldots v_{n-1}\} \mapsto \{\texttt{true}, \texttt{false}\}\) for the \acs{3-SAT} instance.
    First, define a \emph{satisfier function} \(S: \{C_1, \ldots, C_m\} \mapsto \{v_0, \ldots, v_{n-1}, \lnot v_0, \ldots, \lnot v_{n-1}\}\), such that the following conditions hold:
    \begin{itemize}
        \item  $S(C_j)$ is a literal that appears in the clause $C_j$ and
        \item the literal $S(C_j)$ evaluates to \texttt{true} when applying $\phi$ i.e., \(\phi(S(C_j)) = \texttt{true}\).
    \end{itemize}
    Since \(\phi\) is a satisfying assignment, such a function exists.

    For each variable \(v_i\), let \(j, k, \ell\) be the indices of the clauses where \(v_i\) appears positively and negatively, respectively, i.e., \(C^{(i)}_{pos1} = \gamma^{j}\), \(C^{(i)}_{pos2} = \gamma^{k}\) and \(C^{(i)}_{neg\vphantom{p2}} = \gamma^{\ell}\).
    Now, we set the coefficients \(\lambda^{(i)}_k\) for each variable \(v_i\) as follows:

    \textit{Case 1:} (\(\phi(v_i) = \texttt{true}\)).
    Set \(\lambda^{(i)}_1 = 1\) if \(S(C_j) = v_i\) and similarly \(\lambda^{(i)}_2 = 1\) if \(S(C_k) = v_i\).
    Additionally, set \(\lambda^{(i)}_4 = 2 - (\lambda^{(i)}_1 + \lambda^{(i)}_2)\).
    All other coefficients are set to 0.

    \textit{Case 2:} (\(\phi(v_i) = \texttt{false}\)).
    If \(S(C_\ell) = \lnot v_i\), then set \(\lambda^{(i)}_3 = 1\) and \(\lambda^{(i)}_5 = 1\).
    Otherwise, set \(\lambda^{(i)}_4 = 2\).
    All other coefficients are set to 0.

    We must now verify that the constructed linear combination produces the target vector.
    For the first component, note that by construction each clause-paying coefficient (i.e., \(\lambda^{(i)}_1, \lambda^{(i)}_2\), and \(\lambda^{(i)}_3\)) is set to 1 if and only if its corresponding literal was chosen by the satisfier function \(S\).
    Since \(S\) selects exactly one literal for each clause, the sum correctly evaluates to \(\sum_{j=0}^{m-1}{\gamma^j}\).
    For the second component, we have ensured by construction of the linear combination that for each variable \(v_i\), \(\lambda^{(i)}_1 + \lambda^{(i)}_2 + 2\lambda^{(i)}_3 + \lambda^{(i)}_4 = 2\).
    Thus, the sum evaluates correctly to \(\sum_{i=0}^{n-1}{2 \cdot \gamma^i}\).
    The third component counts the usage of slack vectors (i.e.,\(\vect{x^{(i)}_4}, \vect{x^{(i)}_5}\)).
    By construction, \(\lambda^{(i)}_1 + \lambda^{(i)}_2 + \lambda^{(i)}_3 + \lambda^{(i)}_4 +\lambda^{(i)}_5 = 2 \) for each \(i \in [n]_0\), resulting in \(2n\) vector selections in total.
    Since \(m\) of those are clause-paying (i.e.,\(\vect{x^{(i)}_1}, \vect{x^{(i)}_2}, \vect{x^{(i)}_3}\)), the number of selected slack vectors must be \(2n - m\).
    It remains to show the equality for the remaining components (namely \(\hat{\vect \beta}, \hat{\vect \chi}, \hat{\vect y}\)).
    By construction, \(\hat{\vect \beta} \coloneqq \sum_{i=0}^{n-1}{2\cdot \vect \beta(i)}\) where \(\vect \beta(i)\) denotes the vector of unique variables for any fixed \(i\).
    The constructed linear combination chooses exactly two vectors for each variable \(i\), therefore \(\hat{\vect \beta}\) is hit exactly.
    The components \((\chi_1, \chi_2, \chi_3, \chi_4)\) have binary values in \(\vect{x_1^{(i)}}, \ldots, \vect{x_5^{(i)}}\) due to \Cref{thm:possible-solutions}, thus they sum up to some value \(\hat{\vect \chi} \in \{0, \ldots, 2n\}^4\).
    Let now \(\hat{\vect \chi}\) be the correct vector.
    Then \(\hat{\vect y}\) is by construction correct for any combination of \(2n\) solutions to \(\vect s^\intercal\vect x=B\), as argued in the construction of \(\mathcal{I}^{\texttt{BP}}(\hat{\vect \chi})\).

    "$\Leftarrow$": For the "only if" direction, assume that there exists a non-negative integer linear combination of the vectors \(\lambda^{(i)}_k\) that produces the target vector for some \(\hat{\vect \chi} \in \{0, \ldots, 2n\}^4\).

    The satisfying assignment \(\phi\) can now be constructed as follows for each variable \(v_i\):
    \begin{align*}
        \phi(v_i) = \begin{cases}
                        \texttt{false}, & \text{if } \lambda^{(i)}_3 = 1 \\
                        \texttt{true},  & \text{otherwise}
                    \end{cases}
    \end{align*}
    We must now show that \(\phi\) satisfies all clauses.

    We first observe that since \(\gamma > 4n\) there are no carries between the powers of \(\gamma\) in the first two components.
    Indeed, the constraint \(\sum_{i=0}^{n-1}\sum_{k=1}^{5}\lambda^{(i)}_k \leq 2n\) implies \(\lambda^{(i)}_k \leq 2n\) for each \(i \in [n]_0\) and \(k \in [5]\).
    Since the maximum coefficient in the second components is \(2 \cdot \gamma^i\), the maximum contribution to the \(i\)-th power of \(\gamma\) in each component is \(2n \cdot 2 \cdot \gamma^i = 4n \cdot \gamma^i < \gamma^{i+1}\).
    Note that this argument also holds for the first component since in any well-structured \acs{3-SAT} instance, each variable appears in at most three clauses, and thus \(m \leq 3n < 4n < \gamma\).

    The first component of the vector equation is
    \begin{align}
        \sum_{i=0}^{n-1}\left(\lambda^{(i)}_1 C^{(i)}_{pos1} + \lambda^{(i)}_2 C^{(i)}_{pos2} + \lambda^{(i)}_3 C^{(i)}_{neg\vphantom{pos2}}\right) = \sum_{j=0}^{m-1}{\gamma^j}.
        \label{eq:third-component}
    \end{align}
    Remember that the \(C_{pos1}^{(i)}, C_{pos2}^{(i)}, C_{neg\vphantom{pos2}}^{(i)}\) terms are of the form \(\gamma^j\) for some integer \(j\) and all \(\lambda^{(i)}_k\) are non-negative integers.
    Then, for each clause \(j\), there must exist exactly one \(\lambda^{(i)}_k\), \(i \in [n]_0, k \in [3]\) that is equal to 1.
    It remains to show that for each \(i \in [n]_0\): \(\lambda^{(i)}_3 = 1 \Rightarrow \lambda^{(i)}_1 = \lambda^{(i)}_2 = 0\), i.e., if \(v_i\) is set to \texttt{false}, it is only used to satisfy the clause in which it appears negatively.
    To this end, consider the second component of the vector equation
    \begin{align}
        \sum_{i=0}^{n-1}\left(\lambda^{(i)}_1 + \lambda^{(i)}_2 + 2 \cdot \lambda^{(i)}_3 + \lambda^{(i)}_4\right) \cdot \gamma^i = \sum_{i=0}^{n-1}{2 \cdot \gamma^i}.
        \label{eq:first-component}
    \end{align}
    Again, since \(\gamma > 4n\), there are no carries between the powers of \(\gamma\).
    Thus, for each \(i \in [n]_0\), we must have \(\lambda^{(i)}_1 + \lambda^{(i)}_2 + 2 \cdot \lambda^{(i)}_3 + \lambda^{(i)}_4 = 2\).
    This implies that if \(\lambda^{(i)}_3 = 1\), then \(\lambda^{(i)}_1 = \lambda^{(i)}_2 = 0\).
    Therefore, the assignment \(\phi\) is well-defined.
\end{proof}

\subsection{Putting It All Together}
\mainThm*
\begin{proof}
    Let \(\mathcal{I}^{\texttt{SAT}}\) be a well-structured \acs{3-SAT} instance with \(n\) variables.
    We transform \(\mathcal{I}^{\texttt{SAT}}\) into a family of \acs{Bin Packing} instances \(\mathcal{I}^{\texttt{BP}}(\hat{\vect \chi})\) as in \Cref{lem:linear-combination}.
    The encoding length of each instance \(\mathcal{I}^{\texttt{BP}}(\hat{\vect \chi})\) is defined as \(|I| = O(d \cdot \log(B) + d \cdot \log(\|\vect a\|_\infty))\).
    By \Cref{aggilp-dim} we have: \(d = O(\log n)\) item types and bin capacity \(B \leq \gamma^{O(n^2 \log n)}\).
    Since the amounts \(\vect a\) are constructed to be the sum of \(2n\) solutions \(\vect x\) to the \ac{ILP}, we have
    \begin{align*}
        \|\vect a\|_\infty \leq 2n \cdot \|\vect x\|_\infty \leq 2n \cdot \gamma^{O(n^2)} = \gamma^{O(n^2)}.
    \end{align*}
    Thus, with \(\gamma = 4n+1 > 4n\), the encoding length of each constructed \acs{Bin Packing} instance is \(|I| = O(\log n \cdot (n^2 \log(n)^2) + \log n \cdot n^2\log n) = O(n^2 (\log n)^3)\).

    Now suppose there exists an algorithm for \acs{Bin Packing} with a runtime of \({|I|}^{2^{o(d)}}\).
    Then, we can solve \(\mathcal{I}^{\texttt{BP}}(\hat{\vect \chi})\) for all \(\hat{\vect \chi} \in \{0, \ldots, 2n\}^4\) in time \((2n+1)^4 \cdot {|I|}^{2^{o(d)}} = (2n+1)^4 \cdot (n^2 (\log n)^3)^{2^{o(d)}} = (2n+1)^4 \cdot (n^2 (\log n)^3)^{n^{o(1)}}\).
    Using a similar argument as in~\cite{DBLP:conf/sosa/KowalikLMPS24}, we get:
    \begin{align*}
        (2n+1)^4 \cdot (n^2 (\log n)^3)^{n^{o(1)}} \leq 2^{n^{o(1)}3\log(n) + 4\log(2n+1)} \leq 2^{o(n)}
    \end{align*}
    contradicting the \ac{ETH}.
    With this, we have proven \Cref{thm:main-result}.
\end{proof}

\section{Conclusion}
Our result still leaves unresolved the central open problem of whether there exists an \textsc{FPT} algorithm for \acs{Bin Packing} parameterized by $d$.
Goemans and Rothvoss~\cite{DBLP:journals/jacm/GoemansR20} as well as Mnich and van Bevern~\cite{DBLP:journals/cor/MnichB18} posed it as an open problem whether bin packing with $d$ item sizes can be solved in time $ f(d) \cdot O(\log \Delta)^{O(1)}$, where $f$ is an arbitrary function.
Interestingly, an \textsc{FPT} algorithm using at most $OPT +1$ bins is known by Jansen and Solis-Oba~\cite{DBLP:journals/mor/JansenS11}.

Our techniques might be of interest when proving double exponential lower bounds for other high-multiplicity scheduling problems~\cite{DBLP:journals/jacm/GoemansR20} or high-multiplicity n-folds~\cite{DBLP:journals/jacm/GoemansR20,DBLP:journals/orl/KnopKLMO21,DBLP:journals/mp/KnopKLMO23}.
Forall-exist statements are also tightly connected to this topic.
In such problems, we are given a convex set \(\mathcal{Q} \subseteq \mathbb{R}^m\) and an integer matrix \(W \in \mathbb{Z}^{m \times n}\).
A major open problem is whether any algorithm solving \(\forall b\in \mathcal{Q} \cap \mathbb{Z}^m\ \exists x \in \mathbb{Z}^n\) such that \(Wx \leq b\) must necessarily have double-exponential running time in \(n\)~\cite{DBLP:conf/soda/BachERW25}.



\bibliography{src}

\appendix

\label{sec:appendix}
\section{Omitted Proofs}
\label{sec:proofs}

\subsection{Proof of \Cref{def:structured-3sat}}
\label{sec:proof-structuredSat}
\structuredSat*
\begin{proof}
    As shown in~\cite{DBLP:journals/dam/Tovey84} any \acs{3-SAT} instance can be transformed into an instance \(I\) where each variable occurs in at most \(3\) clauses with the addition of at most \(3m\) variables.
    We transform \(I\) into an equivalent instance \(I'\) as follows:

    For any variable \(v_i\), appearing \(k_i = 2\) times, replace its \(j\)-th occurrence by a new variable \(v'_j\).
    To enforce \(v'_1 = v'_2\), add the two clauses:
    \[(v'_1 \lor \lnot v'_2) \land (v'_2 \lor \lnot v'_1). \]
    Now, eliminate any variable \(v_i\) that appears only positively or negatively by satisfying its clauses.
    Finally, each remaining variable appears exactly three times.
    If a variable appears twice negatively and once positively, we replace each occurrence with its literal negation, ensuring it appears exactly twice positively and once negatively.

    In total, this reduction introduces at most \(O(3m)\) variables.
    Therefore, the number of variables in \(I'\) is at most \(n' = O(n)\).
    Since each variable in \(I'\) appears exactly \(3\) times, \(m' = 3n' = O(n)\).
\end{proof}

\subsection[Proof of Linearization]{Proof of \Cref{lem:linearization}}
\label{sec:proof-lemLinearization}
\constraintLinearization*
\begin{proof}
    We aim to prove that \(y = \sum_{j=1}^{k}x_j \cdot \chi_j\) is equivalent to the \(O(k)\) constraints
    \begin{align}
        y - x_j & \leq U \cdot (1 - \chi_j)        & \forall j \in [k] \label{eq:linearization1} \tag{1} \\
        y - x_j & \geq -U \cdot (1 - \chi_j)       & \forall j \in [k] \label{eq:linearization2} \tag{2} \\
        y       & \geq 0                           & \label{eq:linearization3} \tag{3}                   \\
        y       & \leq U \cdot \sum_{j=1}^k \chi_j & \label{eq:linearization4}\tag{4}
    \end{align}
    under the condition \(\sum_{j=1}^{k}\chi_j \leq 1\).

    A key part of this analysis is that due to \(\sum_{j=1}^{k}\chi_j \leq 1\), there are exactly two cases: Either all \(\chi_j = 0\) or exactly one \(\chi_m = 1\) for some \(m \in [k]\) and \(\chi_j = 0\) for all \(j \neq m\).
    We will use these two cases in both directions of the equivalence.

    "$\Rightarrow$": Assume \((\vect y,\vect x,\vect \chi)\) satisfy the nonlinear equation \(y = \sum_{j=1}^{k}x_j \cdot \chi_j\).
    We must now show that they also satisfy the inequalities \eqref{eq:linearization1}-\eqref{eq:linearization4}.

    \noindent\textit{Case 1}: (\(\sum_{j=1}^{k}\chi_j = 0\)).
    This implies \(\chi_j = 0\) for all \(j\), thus \(y = \sum_{j=1}^{k}x_j \cdot \chi_j = 0\).
    In this case the inequalities \eqref{eq:linearization1}-\eqref{eq:linearization4} hold.

    \noindent\textit{Case 2}: (\(\sum_{j=1}^{k}\chi_j = 1\)).
    This implies that exactly one \(\chi_m = 1\) for some \(m\) and \(\chi_j = 0\) for all \(j \neq m\).
    With this, \(y = \sum_{j=1}^{k}x_j \cdot \chi_j = x_m\).
    For \(j \neq m\) the inequalities hold by the same argument as above.
    For \(j = m\), we take a closer look at \eqref{eq:linearization1}:
    \begin{align*}
        y - x_m & \leq U \cdot (1 - \chi_m) \Rightarrow x_m - x_m \leq U \cdot (1 - 1)  \Rightarrow 0 \leq 0
    \end{align*}
    An analogue argument works for \eqref{eq:linearization2}.
    \eqref{eq:linearization3}-\eqref{eq:linearization4} hold due to the given bounds on \(x_m\).

    "$\Leftarrow$": Assume \((\vect y,\vect x,\vect \chi)\) satisfy inequalities \eqref{eq:linearization1}-\eqref{eq:linearization4}.
    We must now show they also satisfy the nonlinear equation \(y = \sum_{j=1}^{k}x_j \cdot \chi_j\).

    \noindent\textit{Case 1}: (\(\sum_{j=1}^{k}\chi_j = 0\)).
    Here the nonlinear equation implies \(y = 0\).
    We must now show that the inequalities \eqref{eq:linearization3}-\eqref{eq:linearization4} also force \(y = 0\).
    With inequality \eqref{eq:linearization4}, we get: \(y \leq U \cdot \sum_{j=1}^k \chi_j = 0\).
    Together with inequality \eqref{eq:linearization3}, we get \(y=0\).

    \noindent\textit{Case 2}: (\(\sum_{j=1}^{k}\chi_j = 1\)).
    his implies that exactly one \(\chi_m = 1\) for some \(m\) and \(\chi_j = 0\) for all \(j \neq m\).
    With this, \(y = \sum_{j=1}^{k}x_j \cdot \chi_j = x_m\).
    Let's look at the inequalities \eqref{eq:linearization1} and \eqref{eq:linearization2} specially for \(j = m\):
    \begin{align*}
        y - x_m \leq U(q - \chi_m) \Rightarrow y - x_m \leq U(1 - 1) \Rightarrow y - x_m \leq 0 \Rightarrow y \leq x_m \\
        y - x_m \geq -U(q - \chi_m) \Rightarrow y - x_m \geq -U(1 - 1) \Rightarrow y - x_m \geq 0 \Rightarrow y \geq x_m
    \end{align*}
    This implies \(y = x_m\).
    Note that all other inequalities are satisfied as argued in the first part in the proof but are not needed to force \(y = x_m\) in this case.

    We do note that constraint \eqref{eq:linearization4} may be omitted if \(\sum_{j=1}^{k}\chi_j = 1\).
\end{proof}

\subsection{Proof of \Crefrange{claim:blocks}{claim:abcd}}
\label{sec:proof-thm-claims}
\claimBlocks*
\begin{proof}
    We show by induction over $j \in [\log(n)+1]_0$, that after $j$ search steps, the variable $z_j$ equals the integer represented by a contiguous subsequence of $\frac{n}{2^j}$ blocks and that this subsequence contains the $i$-th block.

    \emph{Base Case:} Assume $j = 0$. Then, by \Cref{eq:s_logn1} we have, $z_0 = Z$. As $Z$ contains all $\frac{n}{2^j} = n$ blocks, this also holds for the $i$-th block.
    The Euclidean division of \Cref{eq:s_logn2}, in combination with the bound of remainder $r_0$, ensures that both, the quotient $q_0 = \left\lfloor \frac{Z}{(\satbase^m)^{3n/2}}\right\rfloor$ and the remainder $r_0$ are unique. Note, that $Z$ is now split into two integers $q_0$ and $r_0$ that represent two equal-sized subsequences.

    \emph{Inductive Step:} Let $j \in [\log(n)]_0$ and assume that $z_j$ equals the integer represented by a contiguous subsequence of $\frac{n}{2^j}$ blocks and that this subsequence contains the $i$-th block.
    For the same reason as stated in the base case, the Euclidean division gives unique values for $q_j$ and $r_j$.
    \Cref{eq:s_logn4} now simulates a case distinction. Take the $(j+1)$-th bit $x_{j+1}^{\texttt{bin}}$ of the binary representation of $i$. If $x_{j+1}^{\texttt{bin}} = 0$ we keep the lower half, i.e., we set $z_{j+1}=r_j$. If $x_{j+1}^{\texttt{bin}} = 1$ we keep the upper half, i.e., we set $z_{j+1}=q_j$.
    In either case $z_{j+1}$ equals the concatenation of the blocks of the chosen half which contains exactly $\frac{n}{2^{j+1}}$ blocks as $n$ is a power of 2. Also, since $\vect x^{\texttt{bin}}$ is the binary representation of $i$, the $i$-th block remains in the selected half.

    Finally, after $\log(n)$ steps, the subsequence consists of $\frac{n}{2^{\log(n)}} = 1$ block which is the $i$-th block, i.e., \(z_{\log(n)} = C^{(i)}_{pos1} \cdot (\satbase^m)^0 + C^{(i)}_{pos2} \cdot (\satbase^m)^1 + C^{(i)}_{neg\vphantom{p2}} \cdot (\satbase^m)^2\).
\end{proof}

\claimCs*
\begin{proof}
    With \Cref{claim:blocks}, we have \(z_{\log(n)} = C^{(i)}_{pos1} \cdot (\satbase^m)^0 + C^{(i)}_{pos2} \cdot (\satbase^m)^1 + C^{(i)}_{neg\vphantom{p2}} \cdot (\satbase^m)^2\).
    \Cref{eq:cpscn1} together with $c_{pos1} < \satbase^m$ ensures $c_{pos1} = C^{(i)}_{pos1}$ as the Euclidean division separates the quotient $q_c = \left\lfloor \frac{z_{\log(n)}}{\satbase^m}\right\rfloor = C^{(i)}_{pos2} \cdot (\satbase^m)^0 + C^{(i)}_{neg\vphantom{p2}} \cdot (\satbase^m)^1$ and the remainder $c_{pos1} = z_{\log(n)} \mod \satbase^m$. Note that both values are unique for fixed $i$ and $Z$.

    \Cref{eq:cpscn2} together with $c_{pos2} < \satbase^m$ simulates another Euclidean division that now extracts the correct values for $c_{pos2}$ and $c_{neg}$. More concretely, we obtain the unique values $c_{neg} = \left\lfloor\frac{q_c}{\satbase^m}\right\rfloor = C^{(i)}_{neg\vphantom{p2}}$ and $c_{pos2} = q_c \mod \satbase^m = C^{(i)}_{pos2}$.
\end{proof}

\claimAbcd*
\begin{proof}
    Note that by \cite{DBLP:conf/esa/JansenPT25}, we have $\tilde{r}_{\log(n)} = \gamma^i$ for given $i \in [n]_0$.
    We now make a case distinction over the value of \(\alpha_3\). Since it is binary, we consider the following two cases.

    \emph{Case 1:} Assume $\alpha_3 = 1$. Then the left-hand side of \Cref{eq:alpha3} equals 0. This implies $\chi_1+\chi_2+\chi_3 = 0$ and therefore $\chi_1=\chi_2=\chi_3=0$.
    With this, and \Cref{eq:alpha1}, we get:
    \begin{align*}
        \alpha_1 = c_{pos1}\cdot 0 + c_{pos2}\cdot 0 + c_{neg}\cdot 0 = 0
    \end{align*}
    With \Cref{eq:alpha2}, we get
    \begin{align*}
        \alpha_2 = \tilde{r}_{\log(n)} \cdot \chi_4
    \end{align*}
    Now, the inequality \(\chi_4 \leq \alpha_3\) allows \(\chi_4 \in \{0, 1\}\), thus:
    \begin{align*}
        \begin{pmatrix}
            \alpha_1 \\ \alpha_2 \\ \alpha_3
        \end{pmatrix}
        = \begin{cases}
              (0, \tilde{r}_{\log(n)}, 1)^\intercal = (0, \gamma^i, 1)^\intercal, & \text{if } \chi_4=1 \\
              (0, 0, 1)^\intercal = (0, 0, 1)^\intercal,                          & \text{if } \chi_4=0
          \end{cases}
    \end{align*}

    \emph{Case 2:} Assume $\alpha_3 = 0$.
    Now, the left-hand side of \Cref{eq:alpha3} equals 1. Therefore, exactly one of $\chi_1,\chi_2$, and $\chi_3$ equals 1 and the other two equal 0. The inequality \(\chi_4 \leq \alpha_3\) implies \(\chi_4 = 0\).
    With \Cref{eq:alpha1,eq:alpha2}, we get the following possibilities for \(\alpha_1\) and \(\alpha_2\):
    \begin{align*}
        \begin{pmatrix}
            \alpha_1 \vphantom{C^{(i)}_{pos1}} \\ \alpha_2 \vphantom{C^{(i)}_{pos1}}\\ \alpha_3 \vphantom{C^{(i)}_{pos1}}
        \end{pmatrix}
        = \begin{cases}
              (C^{(i)}_{pos1}, \tilde{r}_{\log(n)}, 0)^\intercal = (C^{(i)}_{pos1}, \gamma^i, 0)^\intercal,                                        & \text{if } \chi_1=1 \text{ and } \chi_2=\chi_3=0 \\
              (C^{(i)}_{pos2}, \tilde{r}_{\log(n)}, 0)^\intercal= (C^{(i)}_{pos2}, \gamma^i, 0)^\intercal,                                         & \text{if } \chi_2=1 \text{ and } \chi_1=\chi_3=0 \\
              (C^{(i)}_{neg\vphantom{p2}}, 2 \cdot \tilde{r}_{\log(n)}, 0)^\intercal= (C^{(i)}_{neg\vphantom{p2}}, 2 \cdot \gamma^i, 0)^\intercal, & \text{if } \chi_3=1 \text{ and } \chi_1=\chi_2=0
          \end{cases}
    \end{align*}

    As this case distinction is exhaustive, there are no other possibilities.
    This completes the proof.
\end{proof}

\subsection[Aggregation of the ILP]{Aggregation of the \ac{ILP}}
\label{sec:aggregation}

Using \Cref{lem:aggregation} the \ac{ILP} \(A\vect x=\vect b\) with \(k = O(\log(n))\) constraints and \(d = O(\log(n))\) variables, defined by \Crefrange{eq:lis-constraint-linear}{eq:alpha4-eq} can be aggregated into a single constraint of the form:
\begin{equation}\label{eq:aggregated-ilp}
    \begin{array}{l}
        \sum_{i=1}^k \big(\lisM^{i-1} \sum_{j=1}^d (a_{ij} x_j)\big)
        + \sum_{j=1}^d \big(\lisM^{k+j-1} (x_j + y_j)\big)
        + \lisM^{k+d}\big(\sum_{j=1}^d (x_j+y_j)+y_{d+1}\big) \\
        = \sum_{i=1}^k (\lisM^{i-1} b_i)
        + \sum_{j=1}^d \big(\lisM^{k+j-1} u_j\big)
        + \lisM^{k+d} U,
    \end{array}
\end{equation}
where $U \coloneqq \sum_{j=1}^d u_j$ and $\lisM \coloneqq \Delta U + \max(\|\vect b\|_\infty, \|\vect u\|_\infty) + \Delta + 2$.
We now want to show that the resulting single constraint of the form \(\vect s^\intercal\vect x = B\) satisfies the following lemma:
\ilpAggrDim*
\begin{proof}
    First, note that the original \ac{ILP} has $O(\log(n))$ variables. For the aggregation (See \Cref{lem:aggregation}), we add $O(\log(n))$ slack variables to the slack variable vector $\vect y$. Thus, \(\vect s, \vect x \in \ZZ_{\geq 0}^{O(\log(n))}\).
    Next, consider the size of \(U\) and \(M\).
    All variables in the original \ac{ILP} (see \Cref{sec:equalityConst}) are upper bounded by \(2\gamma^{3nm} = \gamma^{O(nm)} = \gamma^{O(n^2)}\).
    Since \(d = O(\log(n))\), this directly implies \(U \leq O(\log(n)) \cdot \gamma^{O(n^2)} = \gamma^{O(n^2)}\).
    Furthermore, the maximum coefficient \(\Delta\) as well as the right-hand side \(\|\vect b\|_\infty\) are bounded by \(\gamma^{3mn}\) (see \Cref{eq:C5,alpha1-eq} or \eqref{eq:C10} in \Cref{sec:equalityConst}).
    With these bounds for \(\Delta, U\), and \(\|\vect b\|_\infty\), we get \(M = \gamma^{O(n^2)}\).

    Observe that the vector \(\vect s\) lists the coefficients of variables \( x_j, y_j\) in \Cref{eq:aggregated-ilp}.
    The coefficient for any given variable \(x_j\) is at most
    \begin{align*}
        M^{k+d} + M^{k+j-1} + \sum_{i=1}^k a_{ij} M^{i-1} \leq 2\cdot M^{k+d} + \sum_{i=1}^k \Delta M^{i-1} \leq 2 \cdot M^{k+d} + \Delta \cdot M^k \leq O(M^{k+d})
    \end{align*}
    Thus, the maximum coefficient is of magnitude \(\|\vect s\|_{\infty} \leq O(M^{k+d}) = \gamma^{O(n^2 \log(n))}\).
    The same argument can be made for the right-hand side of \Cref{eq:aggregated-ilp}, yielding \(B \leq \gamma^{O(n^2 \log(n))}\).

    The maximum value a variable in the aggregated \ac{ILP} may take is $\gamma^{O(n^2)}$, as each variable in the original \ac{ILP} is bounded by at most \(\gamma^{O(n^2)}\) and the introduced slacks cannot be larger than $U = \gamma^{O(n^2)}$ (see \Cref{lem:aggregation} for details).
    Thus, we have $\|\vect x\|_\infty \leq \gamma^{O(n^2)}$ for any feasible solution \(\vect x\).
\end{proof}

\section{Omitted Constraints}
\label{sec:linearConst}
In this section, we give the linearized constraints referred to in \Cref{thm:possible-solutions}. The constraints can be obtained through a direct application of \Cref{lem:linearization} to \Cref{eq:s_logn4} and \Cref{eq:alpha1,eq:alpha2}.

\subsection[Linearization of Constraints]{Linearization of Constraints (\nlinsymb \(\rightarrow\) \linsymb)}
\label{sec:linConst-decoding}
\subsubsection{Decoding Constraints (\ref{eq:s_logn4})}
The linearization follows by applying \Cref{lem:linearization}.
As an upper bound for each variable, we set $U^{\texttt{dc}}:=\satbase^{3nm}$, since the maximum size of any variable is bounded by \(Z\), which in turn is bounded by \(\satbase^{3nm}\) (see \Cref{lem:sat-encoding} for details):
\begin{align*}
    \begin{array}{rll}\tag{\linsymb C5}\label{eq:decoding-linear}
        z_{j+1} - q_j & \leq U^{\texttt{dc}} \cdot (1 - x^{\texttt{bin}}_{j+1})  & \forall j \in [\log(n)]_0 \\
        z_{j+1} - q_j & \geq -U^{\texttt{dc}} \cdot (1 - x^{\texttt{bin}}_{j+1}) & \forall j \in [\log(n)]_0 \\
        z_{j+1} - r_j & \leq U^{\texttt{dc}} \cdot x^{\texttt{bin}}_{j+1}        & \forall j \in [\log(n)]_0 \\
        z_{j+1} - r_j & \geq -U^{\texttt{dc}} \cdot x^{\texttt{bin}}_{j+1}       & \forall j \in [\log(n)]_0
    \end{array}
\end{align*}
Note that we may omit constraint \eqref{eq:linearization4} in the linearization since \((1 - x^{\texttt{bin}}_{j+1}) + x^{\texttt{bin}}_{j+1} = 1\).

\subsubsection{Clause Constraints (\ref{eq:alpha1}), (\ref{eq:alpha2})}
The linearization follows by applying \Cref{lem:linearization}, where \(U^\texttt{cc}\) is an upper bound for the involved variables.
By construction, \(\alpha_1, c_{p1}, c_{p2}, c_{n} \leq \gamma^m\) and \(\alpha_2 \leq 2 \cdot \gamma^n\).
Therefore, we set $U^{\texttt{cc}} \coloneqq \max(2 \cdot\gamma^n,\gamma^m-1)$:

\begin{align*}
    \begin{array}{rl}\tag{\linsymb C9}\label{eq:alpha1-linear}
        \alpha_1 - c_{pos1} & \leq U^{\texttt{cc}} \cdot (1 - \chi_1)                        \\
        \alpha_1 - c_{pos1} & \geq -U^{\texttt{cc}} \cdot (1 - \chi_1)                       \\
        \alpha_1 - c_{pos2} & \leq U^{\texttt{cc}} \cdot (1 - \chi_2)                        \\
        \alpha_1 - c_{pos2} & \geq -U^{\texttt{cc}} \cdot (1 - \chi_2)                       \\
        \alpha_1 - c_{neg}  & \leq U^{\texttt{cc}} \cdot (1 - \chi_3)                        \\
        \alpha_1 - c_{neg}  & \geq -U^{\texttt{cc}} \cdot (1 - \chi_3)                       \\
        \alpha_1            & \leq U^{\texttt{cc}}\cdot\left(\chi_1 + \chi_2 + \chi_3\right)
    \end{array}
\end{align*}
\begin{align*}
    \begin{array}{rl}\tag{\linsymb C10}
        \alpha_2 - \tilde{r}_{\log(n)}         & \leq U^{\texttt{cc}} \cdot (1 - \chi_1)                                   \\
        \alpha_2 - \tilde{r}_{\log(n)}         & \geq -U^{\texttt{cc}} \cdot (1 - \chi_1)                                  \\
        \alpha_2 - \tilde{r}_{\log(n)}         & \leq U^{\texttt{cc}} \cdot (1 - \chi_2)                                   \\
        \alpha_2 - \tilde{r}_{\log(n)}         & \geq -U^{\texttt{cc}} \cdot (1 - \chi_2)                                  \\
        \alpha_2 - 2 \cdot \tilde{r}_{\log(n)} & \leq U^{\texttt{cc}} \cdot (1 - \chi_3)                                   \\
        \alpha_2 - 2 \cdot \tilde{r}_{\log(n)} & \geq -U^{\texttt{cc}} \cdot (1 - \chi_3)                                  \\
        \alpha_2 - \tilde{r}_{\log(n)}         & \leq U^{\texttt{cc}} \cdot (1 - \chi_4)                                   \\
        \alpha_2 - \tilde{r}_{\log(n)}         & \geq -U^{\texttt{cc}} \cdot (1 - \chi_4)                                  \\
        \alpha_2                               & \leq U^{\texttt{cc}} \cdot\left(\chi_1 + \chi_2 + \chi_3 + \chi_4 \right)
    \end{array}
\end{align*}

\subsection{Construction of Equality Constraints}
\label{sec:equalityConst}
In this section, we introduce slack variables to turn the linear constraints of the prior section into equations.
We start by stating the linear equality constraints equivalent to \Cref{eq:lis-constraint}.

\subsubsection{Variable Constraints}
See~\cite{DBLP:conf/esa/JansenPT25} for details and proofs.
Our notation differs slightly: we renamed some variables and, unlike~\cite{DBLP:conf/esa/JansenPT25}, \(x^{\texttt{bin}}_1\) is the most significant bit in \(\vect x^{\texttt{bin}}\).
\begin{align}
    \label{eq:lis-constraint-linear}
    \begin{array}{rll}\tag{C1}
        \tilde{r}_0                                                                                    & = 1
        \\
        (\gamma^{2^{j+1}}+\gamma^{2^j})\tilde{y}_j - (\gamma^{2^j}+1)\tilde{r}_{j+1} + \tilde{s}_{1,j} & = \gamma^{2^j}   & \forall j \in [\log(n)]_0
        \\
        - \gamma^{2^j}\tilde{y}_j  + \tilde{r}_{j+1} + \tilde{s}_{2,j}                                 & = \gamma^{2^j}-1 & \forall j \in [\log(n)]_0
        \\
        x^{\texttt{bin}}_{\log(n) - j} + \tilde{s}_{3,j}                                               & = 1              & \forall j \in [\log(n)]_0
        \\
        x^{\texttt{bin}}_{\log(n) - j} - \tilde{y}_j + \tilde{s}_{4,j}                                 & = 0              & \forall j \in [\log(n)]_0
        \\
        \tilde{y}_j - (\gamma^{2^j}+1) x^{\texttt{bin}}_{\log(n) - j} + \tilde{s}_{5,j}                & = 0              & \forall j \in [\log(n)]_0
        \\
        (\gamma^{2^j} -1) \tilde{z}_j + \tilde{r}_j - \tilde{r}_{j+1}                                  & = 0              & \forall j \in [\log(n)]_0
        \\
        \gamma^{2^j} x^{\texttt{bin}}_{\log(n) - j} - \tilde{z}_j + \tilde{r}_j + \tilde{s}_{7,j}      & = \gamma^{2^j}   & \forall j \in [\log(n)]_0
        \\
        -\gamma^{2^j} x^{\texttt{bin}}_{\log(n) - j} + \tilde{z}_j + \tilde{s}_{8,j}                   & = 0              & \forall j \in [\log(n)]_0
        \\
        \tilde{z}_j - \tilde{r}_j +\tilde{s}_{9,j}                                                     & = 0              & \forall j \in [\log(n)]_0
        \\
        \tilde{r}_{\log(n)} - s_{9 \log(n)}                                                            & = 1              &
        \\
        \tilde{r}_{\log(n)} + s_{9 \log(n)+1}                                                          & = \gamma^{n-1}   &
    \end{array}
\end{align}
The lower bound of each variable is 0. The largest upper bound of the variables is $\gamma^{O(n)}$. Also, the absolute values of the coefficients and the right hand sides of these constraints are upper bounded by $\gamma^{O(n)}$.

\subsubsection{Decoding Constraints}
The constraints in \Cref{eq:s_logn3,eq:decoding-linear,eq:cpscn3} can be transformed into equality constraints through the introduction of slack variables \(y^{\texttt{dc}}_{\ell, j}\) for all \(\ell \in [8], j \in [\log(n)]_0\).
\begin{align}
    z_0                                             & = Z
    \tag{C2}                                                                                                     \\
    z_j - q_j \cdot (\satbase^m)^{3n/2^{j+1}} - r_j & = 0                            & \forall j \in [\log(n)]_0
    \tag{C3}                                                                                                     \\
    r_j + y^{\texttt{dc}}_{1,j}                     & = (\satbase^m)^{3n/2^{j+1}} -1 & \forall j \in [\log(n)]_0
    \tag{C4}
\end{align}
\begin{align*}
    \begin{array}{rll}\tag{C5} \label{eq:C5}
        z_{j+1} - q_j + U^{\texttt{dc}} \cdot x^{\texttt{bin}}_{j+1} + y^{\texttt{dc}}_{2,j}   & = U^{\texttt{dc}} & \forall j \in [\log(n)]_0
        \\
        - z_{j+1} + q_j + U^{\texttt{dc}} \cdot x^{\texttt{bin}}_{j+1} + y^{\texttt{dc}}_{3,j} & = U^{\texttt{dc}} & \forall j \in [\log(n)]_0
        \\
        z_{j+1} - r_j - U^{\texttt{dc}} \cdot x^{\texttt{bin}}_{j+1} + y^{\texttt{dc}}_{4,j}   & = 0               & \forall j \in [\log(n)]_0
        \\
        z_{j+1} - r_j +U^{\texttt{dc}} \cdot x^{\texttt{bin}}_{j+1} - y^{\texttt{dc}}_{5,j}    & = 0               & \forall j \in [\log(n)]_0
    \end{array}
\end{align*}
\begin{align}
    z_{\log(n)} & = q_c \cdot \satbase^m + c_{pos1}
    \tag{C6}                                            \\
    q_c         & = c_{neg} \cdot \satbase^m + c_{pos2}
    \tag{C7}
\end{align}
\begin{align*}
    \begin{array}{rl}\tag{C8}
        c_{pos1} + y^{\texttt{dc}}_{6} & = \satbase^m-1
        \\
        c_{pos2} + y^{\texttt{dc}}_{7} & = \satbase^m-1
        \\
        c_{neg} + y^{\texttt{dc}}_{8}  & = \satbase^m-1
    \end{array}
\end{align*}
Again, each variable is lower bounded by 0. The upper bounds can be set as follows:
\begin{align*}
    z_j                         & \leq (\satbase^m)^{\frac{3n}{2^j}} & \forall j \in [\log(n)]_0                                 \\
    q_j,r_j                     & \leq (\satbase^m)^{3n/2^{j+1}} -1  & \forall j \in [\log(n)]_0                                 \\
    q_c                         & \leq (\satbase^m)^2                                                                            \\
    c_{pos1}, c_{pos2}, c_{neg} & \leq (\satbase^m)-1                                                                            \\
    y^{\texttt{dc}}_{1,j}       & \leq (\satbase^m)^{3n/2^{j+1}} -1  & \forall j \in [\log(n)]_0                                 \\
    y^{\texttt{dc}}_{\ell,j}    & \leq 2(\satbase^m)^{3n}            & \forall \ell \in \{2,\dots,5\}, \forall j \in [\log(n)]_0 \\
    y^{\texttt{dc}}_{\ell}      & \leq (\satbase^m)-1                & \forall \ell \in \{6,7,8\}
\end{align*}

Both, the largest absolute value of the coefficients and the right-hand side is $U^{\texttt{dc}}=\satbase^{3nm}$ (\Cref{eq:C5}).

\subsubsection{Clause Constraints}
Finally, the constraints in \Crefrange{eq:alpha1-linear}{eq:alpha4-linear} can be transformed into equality constrains through the introduction of slack variables \(y^{\texttt{cc}}_\ell\) for each \(\ell \in [15]\).
\begin{align*}
    \begin{array}{rl}\tag{C9}\label{alpha1-eq}
        \alpha_1 - c_{pos1} +U^{\texttt{cc}} \cdot \chi_1 +y^{\texttt{cc}}_1  & = U^{\texttt{cc}}
        \\
        -\alpha_1 + c_{pos1} +U^{\texttt{cc}} \cdot \chi_1 +y^{\texttt{cc}}_2 & = U^{\texttt{cc}}
        \\
        \alpha_1 - c_{pos2} +U^{\texttt{cc}}\cdot \chi_2  +y^{\texttt{cc}}_3  & = U^{\texttt{cc}}
        \\
        -\alpha_1 + c_{pos2} +U^{\texttt{cc}}\cdot \chi_2 +y^{\texttt{cc}}_4  & = U^{\texttt{cc}}
        \\\
        \alpha_1 - c_{neg} +U^{\texttt{cc}}\cdot \chi_3 +y^{\texttt{cc}}_5    & = U^{\texttt{cc}}
        \\
        -\alpha_1 + c_{neg} +U^{\texttt{cc}}\cdot \chi_3 +y^{\texttt{cc}}_6   & = U^{\texttt{cc}}
    \end{array}
\end{align*}
\begin{align*}
    \begin{array}{rl}\tag{C10}\label{eq:C10}
        \alpha_2 - \tilde{r}_{\log(n)} +U^{\texttt{cc}}\cdot\chi_1 +y^{\texttt{cc}}_7             & = U^{\texttt{cc}}
        \\
        -\alpha_2 + \tilde{r}_{\log(n)} +U^{\texttt{cc}}\cdot\chi_1 +y^{\texttt{cc}}_8            & = U^{\texttt{cc}}
        \\
        \alpha_2 - \tilde{r}_{\log(n)} +U^{\texttt{cc}}\cdot\chi_2 +y^{\texttt{cc}}_9             & = U^{\texttt{cc}}
        \\
        -\alpha_2 + \tilde{r}_{\log(n)} +U^{\texttt{cc}}\cdot\chi_2 +y^{\texttt{cc}}_{10}         & = U^{\texttt{cc}}
        \\
        \alpha_2 - 2 \cdot \tilde{r}_{\log(n)} +U^{\texttt{cc}}\cdot\chi_3 +y^{\texttt{cc}}_{11}  & = U^{\texttt{cc}}
        \\
        -\alpha_2 + 2 \cdot \tilde{r}_{\log(n)} +U^{\texttt{cc}}\cdot\chi_3 +y^{\texttt{cc}}_{12} & = U^{\texttt{cc}}
        \\
        \alpha_2 - \tilde{r}_{\log(n)} +U^{\texttt{cc}}\cdot\chi_4 +y^{\texttt{cc}}_{13}          & = U^{\texttt{cc}}
        \\
        -\alpha_2 + \tilde{r}_{\log(n)} +U^{\texttt{cc}}\cdot\chi_4 +y^{\texttt{cc}}_{14}         & = U^{\texttt{cc}}
    \end{array}
\end{align*}
\begin{align}
    \chi_1 + \chi_2 + \chi_3 + \alpha_3    & =  1
    \tag{C11}                                     \\
    \chi_4 -\alpha_3 +y^{\texttt{cc}}_{15} & = 0
    \tag{C12}\label{eq:alpha4-eq}
\end{align}

As usual, the lower bound of the variables is 0. We can set the upper bounds to:
\begin{align*}
    \alpha_1             & \leq \gamma^m                                             \\
    \alpha_2             & \leq 2\gamma^n                                            \\
    \alpha_3             & \leq 1                                                    \\
    \chi_\ell            & \leq 1                            & \forall \ell \in [4]  \\
    y^{\texttt{cc}}_\ell & \leq 2U^{\texttt{cc}} = 4\gamma^n & \forall \ell \in [14] \\
    y^{\texttt{cc}}_{15} & \leq 1
\end{align*}

Here, the largest absolute value of the coefficients and the right-hand side is $U^{\texttt{cc}} = \max(2 \cdot\gamma^n,\gamma^m-1)$ (\Cref{alpha1-eq,eq:C10}).

\end{document}